\documentclass[11pt, oneside]{article}   	
\usepackage{geometry}                		
\usepackage{graphicx}				
\usepackage{amssymb}

\usepackage{graphicx}
\usepackage{amssymb}
\usepackage{amsmath}
\usepackage{amsfonts}
\usepackage{amssymb}
\usepackage{amsthm}
\usepackage{mathtools}
\usepackage{algorithm}
\usepackage{algcompatible}
\usepackage{thmtools,thm-restate}
\usepackage[affil-it]{authblk}
\usepackage[square,numbers,compress]{natbib}
\usepackage{nicefrac}
\usepackage{comment}
\usepackage{ifpdf}
\usepackage{bm}
\usepackage{color}
\usepackage[displaymath,mathlines]{lineno}
\usepackage{scalerel}
\usepackage{stackengine,wasysym}
\usepackage{hyperref}
\usepackage[stable]{footmisc}
\usepackage{enumerate}
\usepackage{tikz-cd}
\usetikzlibrary{arrows,positioning}
\tikzset{
  shift left/.style ={commutative diagrams/shift left={#1}},
  shift right/.style={commutative diagrams/shift right={#1}}
}
\usepackage{subcaption}
\usepackage{enumitem}
\usepackage{booktabs}

\newtheorem{thm}{Theorem}[section]

\newtheorem{asump}{Assumption}

\newtheorem{rk}{Remark}

\newcommand{\tp}[1]{{#1}^{\mathsf T}}
\renewcommand{\bar}{\overline}
\newcommand{\eps}{\epsilon}
\newcommand{\pa}{\partial}

\renewcommand{\eps}{\varepsilon}
\renewcommand{\epsilon}{\varepsilon}
\renewcommand{\Sigma}{\varSigma}

\newcommand{\E}{\mathrm E}
\newcommand{\V}{\mathrm{Var}}

\newcommand{\tr}{\mathrm{tr}}
\newcommand{\VEC}{\mathrm{vec}}

\DeclareMathAlphabet\mathbfcal{OMS}{cmsy}{b}{n}

\newcommand{\half}{\frac12}

\newcommand{\bv}{{\bf v}}
\newcommand{\bV}{{\bf V}}
\newcommand{\bw}{{\bf w}}
\newcommand{\bW}{{\bf W}}

\newcommand{\bzero}{{\bf 0}}

\newcommand{\bY}{{\bf Y}}
\newcommand{\bx}{{\bf x}}
\newcommand{\bX}{{\bf X}}
\newcommand{\bt}{{\bf t}}

\newcommand{\bM}{{\bf M}}
\newcommand{\bC}{{\bf C}}
\newcommand{\bI}{{\bf I}}
\newcommand{\bL}{{\bf L}}
\newcommand{\GP}{\mathcal{GP}}
\newcommand{\mC}{{\mathcal C}}
\newcommand{\mG}{{\mathcal G}}
\newcommand{\mK}{\mathcal{K}}
\newcommand{\mH}{\mathcal{H}}
\newcommand{\mI}{{\mathcal I}}

\newcommand{\mN}{{\mathcal N}}
\newcommand{\mMN}{\mathcal{MN}}
\newcommand{\mO}{{\mathcal O}}

\newcommand{\mbX}{{\mathbb X}}
\newcommand{\mbY}{{\mathbb Y}}

\newcommand{\mbR}{{\mathbb R}}
\newcommand{\bU}{{\bf U}}

\ifpdf
   \graphicspath{{./figure/PNG/}{./figure/PDF/}{./figure/}}
\else
   \graphicspath{{./figure/EPS/}{./figure/}}
\fi

\title{Bayesian Spatiotemporal Modeling for Inverse Problems}
\author{Shiwei Lan\thanks{slan@asu.edu} \qquad
Shuyi Li \qquad Mirjeta Pasha}
\affil{School of Mathematical \& Statistical Sciences, Arizona State University, \\Tempe, AZ 85287}
\date{}

\begin{document}

\maketitle

\begin{abstract}
Inverse problems with spatiotemporal observations are ubiquitous in scientific studies and engineering applications. In these spatiotemporal inverse problems, observed multivariate time series are used to infer parameters of physical or biological interests. Traditional solutions for these problems often ignore the spatial or temporal correlations in the data (static model), or simply model the data summarized over time (time-averaged model). In either case, the data information that contains the spatiotemporal interactions is not fully utilized for parameter learning, which leads to insufficient modeling in these problems.
In this paper, we apply Bayesian models based on spatiotemporal Gaussian processes (STGP) to the inverse problems with spatiotemporal data and show that the spatial and temporal information provides more effective parameter estimation and uncertainty quantification (UQ). We demonstrate the merit of Bayesian spatiotemporal modeling for inverse problems compared with traditional static and time-averaged approaches using a time-dependent advection-diffusion partial different equation (PDE) and three chaotic ordinary differential equations (ODE). We also provide theoretic justification for the superiority of spatiotemporal modeling to fit the trajectories even it appears cumbersome (e.g. for chaotic dynamics). 
\end{abstract}

\noindent{\bf Keywords:} Spatiotemporal Inverse Problems, Spatiotemporal Gaussian Process, Chaotic Dynamics, Trajectory Fitting, Uncertainty Quantification

\section{Introduction}

Many inverse problems in science and engineering involve large scale spatiotemporal data, typically recorded as multivariate time series. There are examples in fluid dynamics that describes the flow of liquid (e.g. petroleum) or gas (e.g. flame jet) \cite{Baukal_2000}. Other examples include dynamical systems with chaotic behavior prevalent in weather prediction \cite{Lorenz_1963}, biology \cite{Liz_2012}, economics \cite{Brooks_1998} etc. where small perturbation of the initial condition could lead to large deviation from what is observed/calculated in time.
The goal of such inverse problems is to recover the parameters from given observations and knowledge of the underlying physics.
The spatiotemporal information is crucial and should be respected when considering proper statistical models for parameter learning. This is not only of interest in statistics, but also beneficial for practical applications of physics and biology to obtain inverse solutions and UQ more effectively.

Traditional methods for these spatiotemporal inverse problems often ignore the time dependence in the data for a simplified solution \cite{villa2020,cleary2020,lan2022}. They either treat the observed time series statically as independent identically distributed (i.i.d.) observations across times \cite{villa2020,lan2022} (hence we refer to it as ``static" model), or summarize them by taking time average or higher order moments \cite{Morzfeld_2018,cleary2020,Huang_2022} (referred as ``time-averaged" approach). The former is prevalent in Bayesian inverse problems with time series observations \cite{lan2022}. The latter is especially common in parameter learning of chaotic dynamics, e.g. Lorenz systems \cite{Lorenz_1963,cleary2020}, due to their sensitivity to the initial conditions and the system parameters, which in turn causes a rough landscape of the objective function. In both scenarios, the spatiotemporal information is not fully integrated into the statistical modeling.

In this paper, we propose to apply Bayesian methods based on GP to the inverse problems with sptiotemporal data to account for the space-time inter-dependence. This leads to fitting the whole trajectories of the observed data, rather than their statistical summaries, with elaborated models. More specifically, we use the STGP model \cite{cressie2011} to fit the observed multivariate time series in comparison with the static or the time-averaged (for summarized data) models. Theoretically, we justify why the STGP model should be preferred to by investigating their Fisher information, which can be used as a measurement of convexity: STGP renders a more convex likelihood than the other two models and leads to an easier learning of the parameters. We also demonstrate in numerical experiments (Section \ref{sec:numerics}) that the STGP model yields parameter estimates closer to the truth with smaller observation window required, and also provides more reasonable UQ results. Note this implies faster convergence (future work) by the STGP model, which is computationally important because complex ODE/PDE systems are usually expensive to solve.

Spatiotemporal reasoning/modeling was introduced to inverse problems. However, it was either qualitatively applied to specific domains such as functional magnetic resonance imaging (fMRI) \cite{Woolrich_2004}, electroencephalography (EEG) \cite{SIREGAR_1996} and electrocardiography (ECG) \cite{Shcherbakova_2021}, or to a simplified Gauss-linear problem \cite{Long_2011,Ojeda_2019,Conjard_2021,Yang_2017}. Spatiotemporal information was also used to construct prior \cite{Zhang_2005} and regularization \cite{Yao_2016,Pasha_2021}, or to reduce the number of parameters \cite{Echeverria_2009}. However, none of them formulates the spatiotemporal modeling in the general framework of Bayesian inverse problems with spatiotemporal observations.
We summarize the main contributions of this work as follows:
\begin{itemize}[itemsep=0pt]
    \item It formulates a Bayesian modeling framework for inverse problems with spatiotemporal data that includes traditional static and time-averaged methods;
    \item It provides a theoretical justification on why the STGP model is preferable in the spatiotemporal inverse problems;
    \item It numerically demonstrates the advantage of the STGP model in parameter learning and UQ.
\end{itemize}

The rest of the paper is organized as follows: Section \ref{sec:BUQ} reviews the background of Bayesian UQ for inverse problems, with a particular framework named \emph{Calibration-Emulation-Sampling (CES)} \cite{cleary2020,lan2022}. In Section \ref{sec:stip} we generalize the problem setup to include spatiotemporal observations and compare the STGP model (Section \ref{sec:stgp}) with the static model (Section \ref{sec:static}) and the time-averaged model (Section \ref{sec:time-average}). We prove in theorems \ref{thm:convexity1} and \ref{thm:convexity2} that the STGP model can have more convex likelihood than the static and the time-averaged models. Then in Section \ref{sec:numerics} we demonstrate the advantage of the STGP model over the other two traditional approaches with inverse problems involving an advection-diffusion equation and three chaotic dynamics. Finally we conclude with some discussions on future directions in Section \ref{sec:conclusion}.

\section{Background: Bayesian UQ for Inverse Problems}\label{sec:BUQ}

In many inverse problems, we are interested in finding an unknown parameter, $u$ (which could be a function or a vector), given the observed data, $y$. The parameter $u$ usually appears as a quantity of interest in the inverse problem,
e.g. the initial condition of a time-dependent advection-diffusion problem (Section \ref{sec:adif}) or the coefficient vector in the chaotic dynamics (Section \ref{sec:chaotic}).
Let $\mbX$ and $\mbY$ be two separable Hilbert spaces.
A forward mapping $\mG:\mbX \rightarrow \mbY$ from the parameter space $\mbX$ to the data space $\mbY$ (e.g. $\mbY=\mathbb R^m$ for $m\geq 1$) connects $u\in \mbX$ to $y\in \mbY$ as follows:
\begin{equation}\label{eq:forward}
y=\mG(u) + \eta, \qquad \eta\sim \mN(0,\Gamma)
\end{equation}
We can define the following potential function (negative log-likelihood), $\Phi:\mbX\times \mbY\to \mathbb R$, often with $\Gamma=\sigma^2 I$: 
\begin{equation}\label{eq:gauss_nz}
\Phi(u;y) = \frac{1}{2} \Vert y-\mG(u)\Vert^{2}_{\Gamma} = \frac{1}{2} \langle y-\mG(u), \Gamma^{-1} (y-\mG(u)) \rangle
\end{equation}
The forward mapping $\mG$ represents physical laws usually expressed as large and complex ODE/PDE systems that could be highly non-linear.
Therefore repeated evaluations of $\Phi(u;y)$ (and hence $\mG(u)$) are expensive for different $u$'s.

In the Bayesian setting, a prior measure $\mu_0$ is imposed on $u$, independent of $\eta$. 
For example, we could assume a Gaussian prior $\mu_0 = \mathcal N(0,\mathcal C)$ with the covariance $\mathcal C$ being a positive, self-adjoint and trace-class operator on $\mbX$.
Then we can obtain the posterior of $u$, denoted as $\mu(u|y)$, using Bayes' theorem \cite{stuart10,dashti2017}: 
\begin{equation}
\label{eq:Bayes}
\frac{d\mu}{d\mu_0}(u) = \frac{1}{Z}\,\exp(-\Phi(u;y)) \ , \quad \textrm{if} \ 0< Z:=\int_{\mbX} \exp(-\Phi(u;y)) \mu_0(du) < +\infty \ .
\end{equation}
Bayesian UQ for inverse problems involves learning the posterior distribution $\mu(du)$ which often exhibits strongly non-Gaussian behavior, posing significant challenges for efficient inference methods such as Markov Chain Monte Carlo (MCMC).

There are three urging computational challenges in the Bayesian UQ for inverse problems: 1) intensive computation for likelihood evaluations, which requires expensive solving of forward problems; 2) complex (non-Gaussian) posterior distributions; and 3) high dimensionality of the discretized parameter (still denoted as $u$ when there is no confusion from the context).
The latter makes the first two more difficult in the sense that high dimensionality not only makes the forward solutions more expensive, but also challenges the robustness of sampling algorithms.
To address these challenges, an approximate inference framework named Calibration-Emulation-Sampling (CES) has recently been proposed by \cite{cleary2020} and developed by \cite{lan2022}. 
It consists of the following three stages:
\begin{enumerate}[itemsep=0pt]
\item {\bf Calibration}: using optimization-based (ensemble Kalman) algorithms to obtain parameter estimation and collect expensive forward evaluations for the emulation step;
\item {\bf Emulation}: recycling forward evaluations from the calibration stage to build an emulator for sampling;
\item {\bf Sampling}: sampling the posterior approximately based on the emulator, which is much cheaper than the original forward mapping.
\end{enumerate}

CES calibrates the model with ensemble Kalman (EnK) methods \cite{evensen1994,Evensen_1996}. Two algorithms, ensemble Kalman inversion (EKI) \cite{Schillings_2017a,Garbuno-Inigo_2020} and ensemble Kalman sampler (EKS) \cite{Garbuno-Inigo_2020,Garbuno_Inigo_2020}, evolve $J$ ensemble particles $\{u^{(j)}\}_{j=1}^J$ according to the following equations respectively:
\begin{subequations}
\label{eq:enk}
\begin{align}
\textrm{EKI}:&& \frac{d u^{(j)}}{dt} &= \frac1J \sum_{k=1}^J \left\langle \mG(u^{(k)}) - \bar\mG, y - \mG(u^{(j)}) + \sqrt{\Sigma} \frac{dW^{(j)}}{dt} \right\rangle_{\Gamma} (u^{(k)} - \bar u)  \label{eq1:enkf_cont}\\
\textrm{EKS}:&& \frac{d u^{(j)}}{dt} &= \frac1J \sum_{k=1}^J \left\langle \mG(u^{(k)}) - \bar\mG, y - \mG(u^{(j)}) \right\rangle_{\Gamma} (u^{(k)} - \bar u) - C(u) \mC^{-1} u^{(j)} + \sqrt{2C(u)} \frac{dW^{(j)}}{dt} \label{eq2:enks_cont}
\end{align}
\end{subequations}
where $\bar u :=\frac1J\sum_{j=1}^J u^{(j)}$, $\bar \mG:=\frac1J\sum_{j=1}^J \mG(u^{(j)})$, $\Sigma=0$ or $\Gamma$, $\{W^{(j)}\}$ are independent cylindrical Brownian motions on $\mbY$, and $C(u):= \frac1J \sum_{j=1}^J  (u^{(j)} - \bar u) \otimes (u^{(j)} - \bar u)$.
Implemented in parallel, EnK algorithms converge quickly to the optimal parameter with a few (usually hundreds of) ensembles without explicit calculation of gradients. However, due to the collapse of ensembles \cite{Schillings_2017a,Schillings_2017b,deWiljes2018,chada2019}, the sample variance given by $\{u^{(j)}\}_{j=1}^J$ tends to underestimate the actual uncertainty \cite[see Figure 1 in][]{lan2022}.

CES recovers the proper uncertainty by running sampling algorithms based on emulators $\mG^e: \mbX \rightarrow \mbY$ trained on data $\{u_n^{(j)}, \mG(u_n^{(j)})\}_{j=1,n=0}^{J,N}$ that have been collected in the calibration stage.
The emulator can be GP \cite{cleary2020} or neural network (NN), e.g. convolutional NN (CNN) \cite{lan2022}, with the latter being preferred to the former for its computational efficiency and no need to design an optimal training set with controlled size.

Once the emulator is built, CES approximately samples from the posterior with dimension-independent MCMC algorithms based on the emulated likelihood $\Phi^e$ and its gradient $D\Phi^e$ (defined by substituting $\mG$ with $\mG^e$ in \eqref{eq:gauss_nz}) at much lower computational cost. 
A class of dimension-independent algorithms -- 
including \emph{preconditioned Crank-Nicolson (pCN)} \cite{cotter13}, \emph{infinite-dimensional MALA ($\infty$-MALA)} \cite{beskos08}, \emph{infinite-dimensional HMC ($\infty$-HMC)} \cite{beskos11}, and \emph{infinite-dimensional manifold MALA ($\infty$-mMALA) \cite{beskos14} and HMC ($\infty$-mHMC) \cite{beskos2017}} -- are used to overcome the deteriorating mixing time of traditional Metropolis-Hastings algorithms as the dimension of parameter space increases.
These algorithms can all be derived from the following Hamiltonian dynamics on manifold $(\mbX, \mK(u))$:
\begin{equation}\label{eq:mHD}
\frac{d^2u}{dt^2} + \mK(u)\, 
\big\{\, \mathcal{C}^{-1}u + \alpha \nabla\Phi(u) \big\} = 0, \quad \left. \left(v:= \frac{du}{dt}\right)\right|_{t=0} \sim\mathcal{N}(0,\mK(u))\ .
\end{equation}
where $\nabla\Phi(u)$ denotes the Fr\'echet derivative of $\Phi$, and $\mK(u)^{-1} = \mC^{-1} + \beta \mH(u)$
with $\mH(u)$ being chosen as Hessian.
They are implemented by numerically simulating \eqref{eq:mHD} for $I$ steps to generate a proposal $u'$ that is accepted with certain probability.
We have $\infty$-mHMC for $\alpha\equiv 1, \beta\equiv 1$.
With $\alpha\equiv 1, \beta\equiv 0$, it reduces to $\infty$-HMC. If we let $I=1$,
then $\infty$-mHMC reduces to $\infty$-mMALA, which becomes pCN further with $\alpha=0$.

\section{Spatiotemporal Inverse Problems (STIP)}\label{sec:stip}

When the observations are taken from a spatiotemporal process, $y(\bx,t)$, simple Gaussian likelihood function as \eqref{eq:gauss_nz} with $\Gamma=\sigma^2 I$, for example, may not be sufficient to describe the space-time interactions.
To address this issue, we propose to rewrite the data model \eqref{eq:forward} in terms of a GP with spatiotemporal kernel $\Gamma(\bx, t)$:
\begin{equation}\label{eq:forward_proc}
y(\bx, t)=\mG(u)(\bx, t) + \eta(\bx, t), \qquad \eta(\bx, t) \sim \GP(0, \Gamma(\bx, t))
\end{equation}

In practice, the forward model often involves time-dependent PDE, e.g. heat equation and Navier–Stokes equations. Therefore, it is crucial to allow for the spatiotemporal correlations in the statistical analysis of such inverse problems.
Compared to \eqref{eq:forward}, model \eqref{eq:forward_proc} offers a more appropriate definition of the likelihood by incorporating the spatiotemporal structures in the data.

Note, the proposed general framework \eqref{eq:forward_proc} also includes many existing statistical models as special cases. For example, if we define the forward map based on some covariates, $\bX(\bx,t)$,
$\mG(\beta)(\bx, t) = f(\beta, \bX(\bx,t))$,  e.g. $f(\beta, \bX(\bx,t))=\bX(\bx,t) \beta(\bx, t)$,
then \eqref{eq:forward_proc} is simply a regression model.
If we set $\mG(u)(\bx, t) = \bL(\bx, t) u(\bx,t)$ with loading matrix $\bL(\bx, t)$, then \eqref{eq:forward_proc} becomes a latent factor model.

In the following, we will introduce the static (Section \ref{sec:static}) and the time-averaged (Section \ref{sec:time-average}) models and unify them in the framework of STGP model (Section \ref{sec:stgp}).
For the convenience of exposition, we fix some notations in the following.
Denote $\bX:=\{\bx_i\}_{i=1}^I$, $\bt:=\{t_j\}_{j=0}^{J-1}$, and $\bY:=y(\bX, \bt)=\{y(\bx_i, t_j)\}_{i=1,j=0}^{I, J-1}$.
$\bC_*$ is the covariance matrix of the covariance kernel $\mC_*$ restricted on the finite-dimensional discrete space.

\subsection{Static model}\label{sec:static}
In the literature of Bayesian inverse problems, the noise $\eta$ is often assumed i.i.d. over time in \eqref{eq:forward_proc}, i.e. $\eta(\bx, t_j)\overset{iid}{\sim}\mN(0, \mC_\bx)$. This leads to the following static model where the temporal correlation is ignored:
\begin{equation}\label{eq:static}
\begin{aligned}
y(\bx, t) | u, \Gamma & \sim \GP(\mG(u)(\bx, t), \Gamma(\bx, t)) \\
\textrm{static}: \qquad \Gamma(\bx, t) &= \mC_\bx \otimes \mI_t\\
\end{aligned}
\end{equation}
where $\mI_t$ is the Dirac operator such that $\mI_t(t, t')=1$ only if $t=t'$.
When the spatial dependence is also suppressed (as in the advection-diffusion example of Section \ref{sec:adif} and in \cite{villa2020, lan2022}), we have $\mC_\bx=\sigma^2_\eps \mI_\bx$.

\begin{figure}[t]
\includegraphics[width=1\textwidth,height=.3\textwidth]{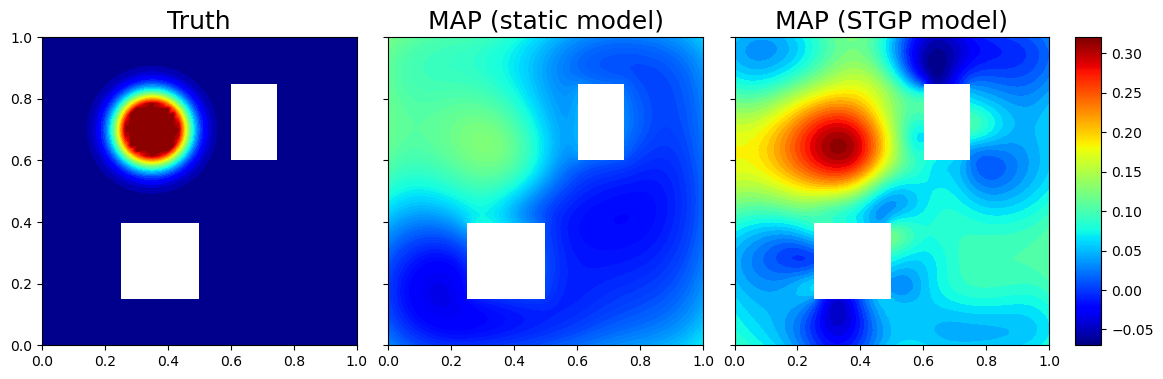}
\caption{Advection-diffusion inverse problem: comparing maximum \emph{a posteriori} (MAP) estimates of parameter $u_0=u(\bx, 0)$ by the static model (middle) and the STGP model (right) with the truth $u_0^\dagger$ (left).}
\label{fig:adif_comparelik}
\end{figure}

Temporal correlation is disregarded in the static model \eqref{eq:static}. 
When there is (spatio-)temporal effect in the residual $\eta$, the static model \eqref{eq:static} may be insufficient to account for the spatiotemporal relationships contained in the data.
For illustration, we consider an inverse problem involving advection-diffusion (Section \ref{sec:adif}) equation \cite{villa2020,lan2022} of an evolving concentration field $u(\bx, t)$, e.g., temperature for heat transfer,
and seek the solution to the initial condition, $u_0=u(\bx, 0)$, based on spatiotemporal solutions observed (through an observation operator $\mO$) on the boundaries of two boxes (Figure \ref{fig:adif_comparelik}, left panel) for a given time period, i.e. $y = \mO u(\bx, t) + \eta, \; \eta \sim N(0,\sigma^2_\eta)$.
As shown in Figure \ref{fig:adif_comparelik}, the simple static model \eqref{eq:static} used in \cite{lan2022} does not account for space-time interactions hence yields the result underestimating the true function $u_0^\dagger$ (left panel).
On the contrary, the estimate by the spatiotemporal model \eqref{eq:STGP} (right panel) is much closer to the truth. 

\subsection{Time-averaged model}\label{sec:time-average}
In many chaotic dynamics, we observe the trajectories as multivariate time series that are very sensitive to the initial condition and the parameters. This usually results in a complex objective function with multiple local minima \cite{Abarbanel_2013}. They in turn form a rough landscape of the objective and pose extreme difficulties on parameter learning \cite{cleary2020} (See also Figure \ref{fig:lrz_pairpdf}). The time-averaged approach is commonly used in the spirit of extracting sufficient statistics from the raw data \cite{Fisher_1922}.

We consider the same data model as in \eqref{eq:forward_proc} with $\mG(u)$ being the observed solution $\bx(t; u, \bx_0)$ of the following chaotic dynamics ($r$-th order ODE) for a given parameter $u\in \mbR^p$:
\begin{equation}\label{eq:chaotic}
\dot \bx := \frac{d\bx}{dt} = f(t, \bx, \bx^{(1)},\cdots,\bx^{(r)}; u), \qquad
\bx(0) = \bx_0 \in \mbR^I
\end{equation}
That is, $\mG(u) = \mO \bx(t; u, \bx_0)$ with an observation operator $\mO$. At each time $t$, the observed vector could include components of $\bx$ and up to their $k$-th order interactions for $k\geq 1$. For example, if $\bx=[x_1,\cdots, x_I]$, we could include all the first and second order terms in the observation vector,
$\mO\bx = [x_1, \cdots, x_I, x_1^2, x_1x_2,\cdots, x_ix_j,\cdots, x_I^2]$.
Because the trajectories of $\mG(u)$ are usually complex, it is often to average them over time and consider the following forward mapping instead:
\begin{equation}
    \mG_T(u; \bx_0) := \frac{1}{T} \int_{t_0}^{t_0+T} \mO \bx(t; u, \bx_0) dt
\end{equation}
where $t_0$ is the spin-up time and $T$ is the window length for averaging the observed trajectories of the dynamics.

Following \cite{cleary2020}, we make the same assumption regarding the dynamical system \eqref{eq:chaotic}:
\begin{asump}
\begin{enumerate}
    \item For $u\in\mbX$, \eqref{eq:chaotic} has a compact attractor $\mathcal A$, supporting an invariant measure $\mu(d\bx; u)$. The system is ergodic, and the following limit of Law of Large Numbers (LLN) is satisfied: for $\bx_0\sim \mu(\cdot; u)$ fixed, with probability one,
    \begin{equation}
        \lim_{T\to \infty}\mG_T(u; \bx_0) = \mG_\infty(u):= \int_{\mathcal A} \mO \bx(t; u, \bx_0) \mu(d\bx; u)
    \end{equation}
    \item The Central Limit Theorem (CLT) holds quantifying the ergodicity: for $\bx_0\sim \mu(\cdot; u)$,
    \begin{equation}\label{eq:gaussian_asump}
       \mG_T(u; \bx_0) \overset{\cdot}{\sim} \mN\left( \mG_\infty(u), T^{-1}\Sigma(u) \right)
    \end{equation}
\end{enumerate}
\end{asump}

The limit $\mG_\infty(u)$ becomes independent of the initial condition $\bx_0$.
However, the finite-time truncation in $\mG_T(u; \bx_0)$, with different random initializations $\bx_0$, generates random errors from the limit $\mG_\infty(u)$, which are assumed approximately Gaussian.
Assume the data $y$ can be observed with a true parameter $u^\dagger$, i.e. $y=\mG_T(u^\dagger; \bx_0)$.
The following time-averaged model is usually adopted for the inverse problems involving chaotic dynamics \cite{cleary2020}:
\begin{equation}\label{eq:time-average}
\begin{aligned}
    y | u, \Sigma(u) &\sim \mN(\mG_\infty(u), T^{-1}\Sigma(u)) \\
    \text{time-average}: \qquad T^{-1}\Sigma(u) &\approx \Gamma_\textrm{obs}
\end{aligned}
\end{equation}
where the empirical covariance $\Gamma_\textrm{obs}$ can be estimated with $\mG_\tau(u; \bx_0)$ for $\tau\gg T$.

In practice, we replace $\mG_\infty(u)$ with $\mG_T(u; \bx_0)$ in \eqref{eq:time-average} and define the potential $\Phi_\textrm{\tiny T}(u)$ of parameter $u$ for the time-averaged model \eqref{eq:time-average} as follows:
\begin{equation}
    \Phi_\textrm{\tiny T}(u) = \half \Vert y- \mG_T(u; \bx_0) \Vert_{\Gamma_\textrm{obs}}^2
\end{equation}
If we observe the trajectories (without component interaction terms, i.e. $\mO\bx=\bx$) at discrete time points $\bt$ with $t_{J-1}=t_0+T$, then $\mO\bx(t; u)$ yields multivariate time series, denoted as $\bX(u)_{I\times J}=\bx(\bt; u)=[\bx(t_0;u),\cdots,\bx(t_{J-1};u)]$. Then we have
\begin{equation}
\mG_T(u; \bx_0) = \bar{\bX}(u) := \bX(u) \frac{\bm{1}_J}{J}, \quad
y = \bX(u^\dagger) \frac{\bm{1}_J}{J}, \quad
    \Gamma_\textrm{obs} = \bX(u^\dagger) \left[\bI_J - \frac{\bm{1}_J \tp{\bm{1}}_J}{J}\right] \tp{\bX(u^\dagger)}
\end{equation}
Denote $\bX_0=\bX(u)-\bX(u^\dagger)$. Therefore the potential $\Phi_\textrm{\tiny T}$ becomes
\begin{equation}
    \Phi_\textrm{\tiny T}(u) = \half \frac{\tp{\bm{1}}_J}{J} \tp{\bX}_0 \Gamma_\textrm{obs}^{-1} \bX_0 \frac{\bm{1}_J}{J} = \half \tr\left[ \frac{\bm{1}_J \tp{\bm{1}}_J}{J^2} \tp{\bX}_0 \Gamma_\textrm{obs}^{-1} \bX_0 \right]
\end{equation}

Note, averaging the trajectories over time does not ease the difficulty of rough landscape, see for instance Figure \ref{fig:lrz_pairpdf} for an illustration. However, the potential function for the following STGP model \eqref{eq:STGP} is more convex around the true values $u^\dagger$ compared with the time-averaged approach \eqref{eq:time-average}.

The aforementioned two approaches, the static model \eqref{eq:static} and the time-averaged model \eqref{eq:time-average}, can be recognized as special cases of a more general framework of spatiotemporal modeling based on STGP, to be discussed in the following section.

\subsection{Spatiotemporal GP model}\label{sec:stgp}
For the spatiotemporal data $y(\bx, t)$ in the inverse problems, we consider the following likelihood model based on STGP:
\begin{equation}\label{eq:STGP}
\begin{aligned}
y(\bx, t) | u, \Gamma & \sim \GP(\mG(u)(\bx, t), \Gamma(\bx, t)) \\
\textrm{STGP}:\qquad \Gamma(\bx, t) &= \mC_\bx \otimes \mC_t
\end{aligned}
\end{equation}
where $\mC_\bx$ and $\mC_t$ are spatial and temporal kernel respectively.

If we observe the process $y(\bx, t)$ according to \eqref{eq:STGP}, the resulted data matrix $\bY=\mG(u)(\bX, \bt)$ follows the matrix normal distribution (denoted as `$\mMN$') \cite{Gupta_2018} for which we can also specify the above-mentioned three models
\begin{subequations}\label{eq:MatN}
\begin{align}
\bY\, |\, \bM, \bU, \bV &\sim \mMN(\bM, \bU, \bV), \quad \bM = \mG(u^\dagger)(\bX, \bt) \notag \\
\textrm{static}: \qquad & \bU_\textrm{\tiny S} = \sigma^2_\eps \bI_\bx, \quad \bV_\textrm{\tiny S} = \bI_t \label{eq1:cov_static} \\
\textrm{time-average}:\qquad & \bU_\textrm{\tiny T} = \Gamma_\textrm{obs}, \quad \bV_\textrm{\tiny T} = J^2 (\bm{1}_J \tp{\bm{1}}_J)^- \label{eq2:cov_Tavg} \\
\textrm{STGP}: \qquad & \bU_\textrm{\tiny ST} = \bC_\bx, \quad \bV_\textrm{\tiny ST} = \bC_t \label{eq3:cov_stgp}
\end{align}
\end{subequations}
where $\bY=\mO\bx(t; u)=\bX(u)$ for the static model and $M^-$ is the pseudo-inverse of $M$.


In all the above three models \eqref{eq:MatN}, we assume $\bY$ i.i.d. over $u$'s.
Denote $\Phi_*$ and $\mI_*$ as potential function and Fisher information matrix with $*$ being `S' for the static model \eqref{eq1:cov_static}, `T' for the time-averaged model \eqref{eq2:cov_Tavg} and `ST' for the STGP model \eqref{eq3:cov_stgp} respectively. 
The following theorem compares the convexity of their likelihoods and indicates that the STGP model \eqref{eq3:cov_stgp} with proper configuration has the advantage of parameter learning with the most convex likelihood among the three models.

\begin{thm}\label{thm:convexity1}
If we set the maximal eigenvalues of $\bC_\bx$ and $\bC_t$ such that $\lambda_{\max}(\bC_\bx)\lambda_{\max}(\bC_t)\leq \sigma^2_\eps$, then the following inequality holds regarding the Fisher information matrices, $\mI_\textrm{\tiny S}$ and $\mI_\textrm{\tiny ST}$, of the static model and the STGP model respectively:
\begin{equation}
    \mI_\textrm{\tiny ST}(u) \geq \mI_\textrm{\tiny S}(u)
\end{equation}
If we control the maximal eigenvalues of $\bC_\bx$ and $\bC_t$ such that $\lambda_{\max}(\bC_\bx)\lambda_{\max}(\bC_t)\leq J\lambda_{\min}(\Gamma_\textrm{obs})$, then the following inequality holds regarding the Fisher information matrices, $\mI_\textrm{\tiny T}$ and $\mI_\textrm{\tiny ST}$, of the time-averaged model and the STGP model respectively:
\begin{equation}
    \mI_\textrm{\tiny ST}(u) \geq \mI_\textrm{\tiny T}(u)
\end{equation}
\end{thm}

\begin{proof}
See Appendix \ref{apx:proof_convexity}.
\end{proof}

The following theorem considers a special case, $\bC_\bx=\Gamma_\textrm{obs}$, under milder condition in comparing the likelihood convexity of the time-averaged model and the STGP model.
\begin{thm}\label{thm:convexity2}
If we choose $\bC_\bx=\Gamma_\textrm{obs}$ and require the maximal eigenvalue of $\bC_t$, $\lambda_{\max}(\bC_t)\leq J$, then the following inequality holds regarding the Fisher information matrices, $\mI_\textrm{\tiny T}$ and $\mI_\textrm{\tiny ST}$, of the time-averaged model and the STGP model respectively:
\begin{equation}
    \mI_\textrm{\tiny ST}(u) \geq \mI_\textrm{\tiny T}(u)
\end{equation}
\end{thm}

\begin{proof}
See Appendix \ref{apx:proof_convexity}.
\end{proof}

\begin{rk}
In general, $\Phi_*(u)$ is not the potential of a Gaussian distribution because of the possible non-linearity of $\mG(u)$. Theorems \ref{thm:convexity1} and \ref{thm:convexity2} indicate that for each $u\in\mbX$, the STGP model can have a more convex Gaussian proxy in the Laplace approximation.
\end{rk}

\begin{rk}
If we view Fisher information as a measurement of (statistical) convexity, the above theorems \ref{thm:convexity1} and \ref{thm:convexity2} indicate that the STGP model can have a likelihood more convex around the true parameter value than either the static model or the time-averaged model does. This implies that parameter learning method based on the STGP model could be more effective in the sense that it may converge faster. 
\end{rk}

Often we are interested in predicting the underlying process $y(\bx, t)$ at future time $t_*$ given the spatiotemporal observations $\bY$. Based on the STGP model \eqref{eq:STGP}, we could use the following posterior predicative distribution
\begin{equation}
p(y(\bx, t_*)|\bY) = \int p(y(\bx, t_*)| u, \bY) p (u|\bY) du
\end{equation}

Denote the conditional prediction $\E[y(\bx, t_*)\vert u,\bY]$ as 
\begin{equation}
    \mG^*(u)(\bx, t_*)=\underbrace{\mG(u)(\bx, t_*)}_{Physical} + \underbrace{\Gamma_{t_*\bt}\Gamma_{\bt\bt}^{-1}(\bY - \mG(u)(\bX, \bt))}_{Statistical}
\end{equation}
Then we predict $y(\bx, t_*)$ with the following predicative mean
\begin{equation}
    \E[y(\bx, t_*)|\bY] = \E_{u|\bY}[\E_{y_*|u,\bY}[y(\bx, t_*)]]
    = \E_{u|\bY}[\mG^*(u)(\bx, t_*)]
    \approx \bar{\mG}(\bx, t_*) + \Gamma_{t_*\bt}\Gamma_{\bt\bt}^{-1}(\bY - \bar{\mG}(\bX, \bt))
\end{equation}
where $\bar{\mG}(\bx, t_*):=\frac{1}{S} \sum_{s=1}^S \mG(u^{(s)})(\bx, t_*)$ with $u^{(s)} \sim p (u|\bY)$.
And we can quantify the uncertainty using the law of total conditional variance:
\begin{equation}
\begin{aligned}
\V[y(\bx, t_*)|\bY] &= \E_{u|\bY}[\V_{y_*|u,\bY}[y(\bx, t_*)]] + \V_{u|\bY}[\E_{y_*|u,\bY}[y(\bx, t_*)]]\\
&= \Gamma_{t_*t_*} - \Gamma_{t_*\bt} \Gamma^{-1}_{\bt\bt} \Gamma_{\bt t_*} + \V_{u|\bY}[\mG^*(u)(\bx, t_*)] \\
&\approx \Gamma_{t_*t_*} - \Gamma_{t_*\bt} \Gamma^{-1}_{\bt\bt} \Gamma_{\bt t_*} 
+ s^2_{\mG^*}(\bx, t_*)
\end{aligned}
\end{equation}
where $s^2_{\mG^*}(\bx, t_*):=\frac{1}{S} \sum_{s=1}^S [\mG^*(u^{(s)})(\bx, t_*)- \bar{\mG^*}(\bx, t_*)]^2$ with $u^{(s)} \sim p (u|\bY)$.

Assume $t_*\notin \bt$.
For static model \eqref{eq:static}, we have $\Gamma_{t_*\bt}=0$ thus $\mG^*(u)(\bx, t_*)=\mG(u)(\bx, t_*)$. Therefore we have the simplified results
\begin{equation}
    \E[y(\bx, t_*)|\bY] \approx \bar{\mG}(\bx, t_*), \qquad \V[y(\bx, t_*)|\bY] \approx \sigma^2_\eps + s^2_{\mG}(\bx, t_*)
\end{equation}
This may underestimate the uncertainty compared with the more general STGP model \eqref{eq:STGP}.
If we are only interested in predicting the forward map $\mG(u)$ to new time $t=t_*$, we actually have similar results
\begin{equation}\label{eq:fwd_pred}
    \E[\mG(u)(\bx, t_*)|\bY] \approx \bar{\mG}(\bx, t_*), \qquad \V[\mG(u)(\bx, t_*)|\bY] \approx s^2_{\mG}(\bx, t_*)
\end{equation}

Note all the above prediction is feasible only if we are able to solve ODE/PDE systems to time $t_*$, i.e. we can evaluate $\mG(u^{(s)})(\bx, t)$ at $t=t_*$. 
When we do not have the computer codes available for doing so, we could model $\mG(u)(\bx, t)$ with another GP $\GP(0, \Gamma^\mG)$ and further predict the forward mapping:
\begin{equation}
\mG(u)(\bx, t_*)| \mG(u)(\bX, \bt) \sim \mN(\Gamma^\mG_{t_*\bt} (\Gamma^\mG_{\bt\bt})^{-1} \mG(u)(\bX, \bt),  \Gamma^\mG_{t_*t_*} - \Gamma^\mG_{t_*\bt} (\Gamma^\mG_{\bt\bt})^{-1} \Gamma^\mG_{\bt t_*} )
\end{equation}

\section{Numerical Experiments}\label{sec:numerics}

In this section, we demonstrate the numerical advantage of spatiotemporal modeling in parameter estimation and UQ.
More specifically, we compare the STGP model \eqref{eq:STGP} with the static model \eqref{eq:static} using an advection-diffusion inverse problem (Section \ref{sec:adif}) previously considered in \cite{villa2020,lan2022} with the static method.
Then we compare the STGP model \eqref{eq:STGP} with the time-averaged model \eqref{eq:time-average} using three chaotic dynamical inverse problems (Section \ref{sec:chaotic}) of which the Lorenz problem (Section \ref{sec:Lorenz}) was studied by \cite{cleary2020} with the time-averaged approach.
Numerical evidences are presented to support that the STGP model \eqref{eq:STGP} is preferable to the other two models.
All the computer codes are publicly available at \url{https://github.com/lanzithinking/Spatiotemporal-inverse-problem}.

\subsection{Advection-diffusion inverse problem}\label{sec:adif}
In this section, we consider an inverse problem governed by a parabolic PDE within the Bayesian inference framework. 
The underlying PDE is a time-dependent advection-diffusion equation that can be applied to heat transfer, air pollution, etc. The inverse problem involves inferring an unknown initial condition $u_0\in L^2(\Omega)$ from spatiotemporal point measurements $\{y(\bx_i, t_j)\}$.

\begin{figure}[t]
\begin{subfigure}[b]{.5\textwidth}
\includegraphics[width=1\textwidth,height=1\textwidth]{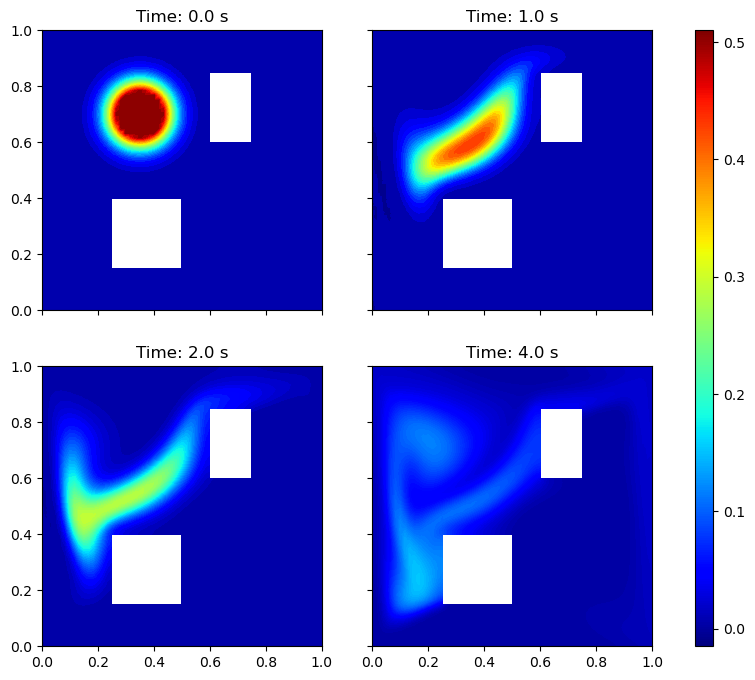}
\caption{True initial condition (top left), and the solutions $u(s)$ at different time points.}
\label{fig:time_soln}
\end{subfigure}
\hspace{2pt}
\begin{subfigure}[b]{.5\textwidth}
\includegraphics[width=1\textwidth,height=1\textwidth]{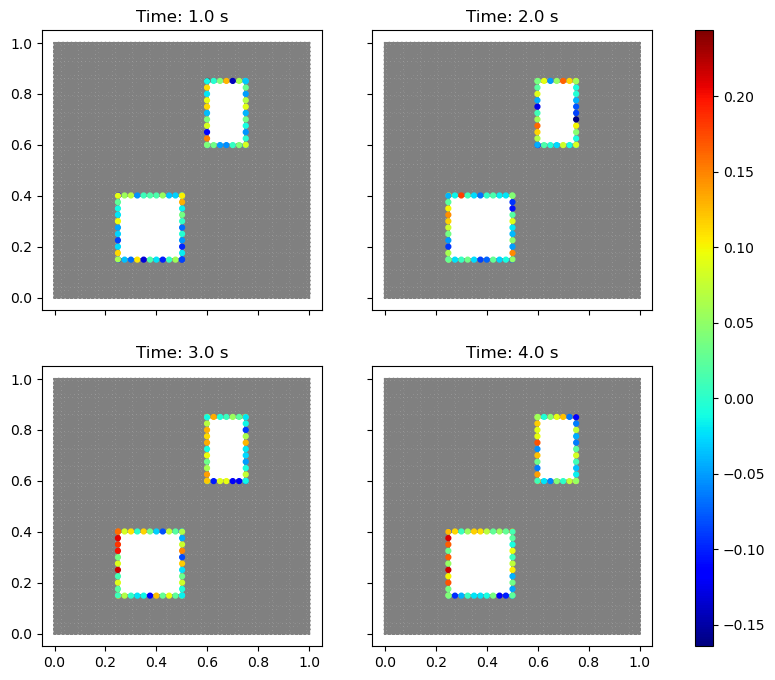}
\caption{Spatiotemporal observations at $80$ selected locations (color dots) across different time points.}
\label{fig:spatiotemporal_obs}
\end{subfigure}
\vspace{-20pt}
\caption{Advection-diffusion inverse problem.}
\end{figure}

The parameter-to-observable forward mapping $\mG : u_0 \to \mO u$ maps the initial condition $u_0$ to pointwise spatiotemporal observations of the concentration field $u(\bx, t)$ through the solution of the following advection-diffusion equation \citep{Petra2011,villa2020}:
\begin{equation}\label{eq:adif}
\begin{aligned}
u_t - \kappa \Delta u + \bv \cdot \nabla u &= 0 \quad in\; \Omega\times (0, T)\\
u(\cdot, 0) &= u_0 \quad in\; \Omega \\
\kappa \nabla u\cdot \vec n &= 0, \quad on\; \pa\Omega\times (0, T)
\end{aligned}
\end{equation}
where $\Omega \subset [0,1]^2$ is a bounded domain shown in Figure \ref{fig:time_soln}, $\kappa=10^{-3}$ is the diffusion coefficient, and $T>0$ is the final time.
The velocity field $\bv$ is computed by solving the following steady-state Navier-Stokes equation with the side walls driving the flow \citep{Petra2011}:
\begin{equation}\label{eq:adif_v}
\begin{aligned}
-\frac{1}{\mathrm{Re}} \Delta \bv + \nabla q + \bv \cdot \nabla \bv &= 0 \quad in\; \Omega\\
\nabla \cdot \bv &= 0 \quad in\; \Omega \\
\bv &={\bf g}, \quad on\; \pa\Omega
\end{aligned}
\end{equation}
Here, $q$ is the pressure, and $\mathrm{Re}$ is the Reynolds number, which is set to 100 in this example. The Dirichlet boundary data ${\bf g}\in \mbR^2$ is given by ${\bf g}={\bf e}_2=(0,1)$ on the left wall of the domain, ${\bf g}=-{\bf e}_2$ on the right wall, and ${\bf g}={\bf 0}$ everywhere else.


\begin{figure}[tbp]
\begin{subfigure}[b]{1\textwidth}
\includegraphics[width=1\textwidth,height=.5\textwidth]{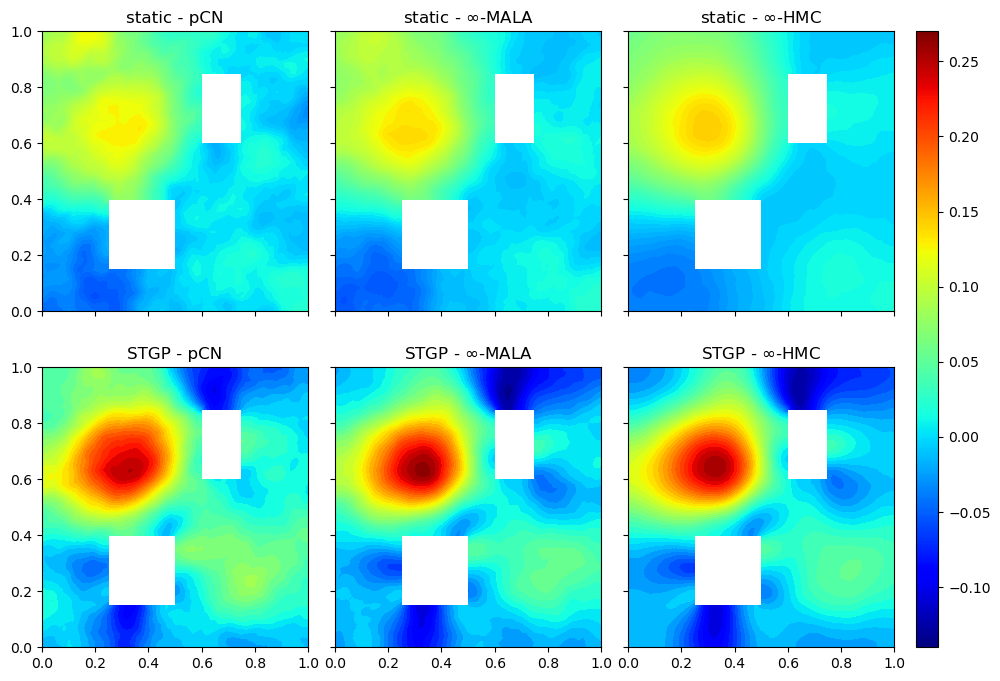}
\caption{Posterior mean estimates of the initial concentration field $u_0(\bx)$.}
\label{fig:adif_postmean_comparelik}
\end{subfigure}
\begin{subfigure}[b]{1\textwidth}
\includegraphics[width=1\textwidth,height=.5\textwidth]{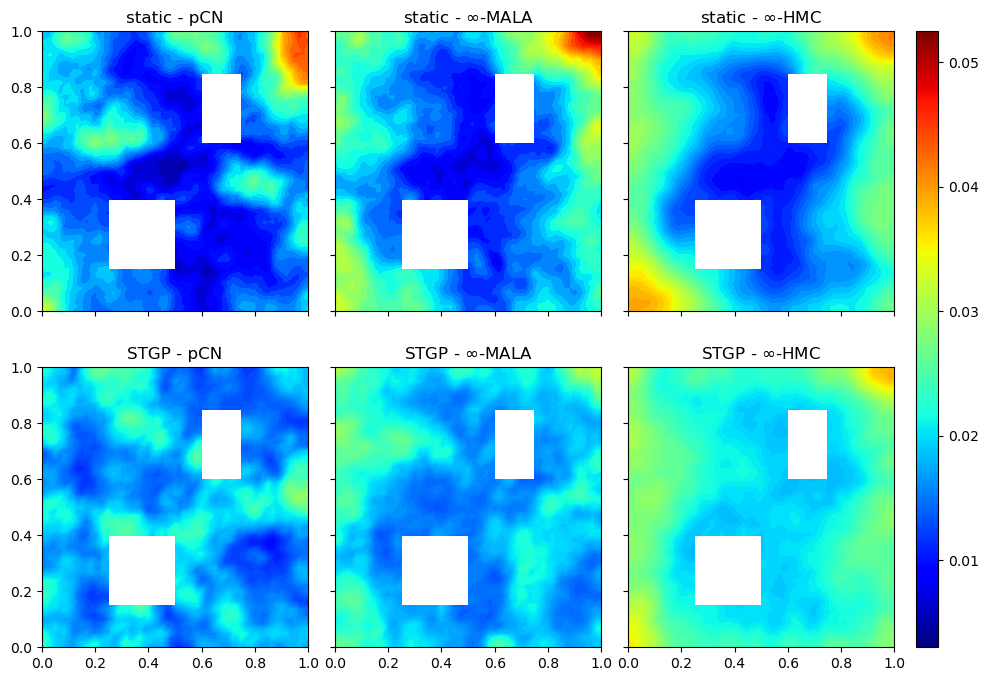}
\caption{Posterior standard deviation estimates of the initial concentration field $u_0(\bx)$.}
\label{fig:adif_poststd_comparelik}
\end{subfigure}
\vspace{-20pt}
\caption{Advection-diffusion inverse problem: comparing posterior estimates of parameter $u_0$ in the static model (upper row) and the STGP model (lower row) based on $5000$ samples by various MCMC algorithms.}
\end{figure}

We set the true initial condition $u_0^\dagger = 0.5\wedge \exp\{-100[(x-0.35)^2+(y-0.7)^2]\}$, illustrated in the top left panel of Figure \ref{fig:time_soln}, which also shows a few snapshots of solutions $u(\bx, t)$ at other time points on a regular grid mesh of size $61\times 61$.
To obtain spatiotemporal observations $\{y(\bx_i, t_j)\}$, we collect solutions $u(\bx, t)$ solved on a refined mesh at $I=80$ selected locations $\{\bx_i\}_{i=1}^I$ across $J=16$ time points $\{t_j\}_{j=1}^J$ evenly distributed between $1$ and $4$ seconds (thus denoted as $\mO u$) and inject some Gaussian noise $\mN(0, \sigma^2_\eta)$ such that the relative noise standard deviation is $\sigma_\eta/\max \mO u = 0.5$, i.e.,
\begin{equation*}
y = \mG(u_0^\dagger) = \mO u(\bx, t; u_0^\dagger) + \eta, \quad \eta \sim \mathcal N(0, \sigma_\eta^2 I_{1280})
\end{equation*}
Figure \ref{fig:spatiotemporal_obs} plots 4 snapshots of these observations at 80 locations along the inner boundary.
In the Bayesian setting, we adopt a GP prior for $u_0\sim \mu_0 = \GP(0, \mC)$ with the covariance kernel $\mC=(\delta \mI -\gamma \Delta )^{-2}$ defined through the Laplace operator $\Delta$,
where $\delta$ governs the variance of the prior and $\gamma/\delta$ controls the correlation length. We set $\gamma=2$ and $\delta=10$ in this example.

The Bayesian inverse problem involves obtaining an estimate of the initial condition $u_0$ and quantifying its uncertainty based on the $80\times 16$ spatiotemporal observations.
The Bayesian UQ in this example is especially challenging not only because of its large dimensionality (3413) of spatially discretized $u$ (Lagrange degree 1) at each time $t$, but also due to the spatiotemporal correlations in these observations.


\begin{table}[ht]\small
\centering
\begin{tabular}{l|lll|lll}
\toprule
 & \multicolumn{3}{|c|}{Estimation} & \multicolumn{3}{c}{Prediction} \\
\cmidrule{2-4} \cmidrule{5-7}
Models &          pCN & $\infty$-MALA & $\infty$-HMC  &  pCN & $\infty$-MALA & $\infty$-HMC \\
\midrule
static & 0.83 (0.023) &  0.81 (0.011) & 0.79 (0.005)  &  0.43 (0.013) &   0.4 (0.006) &  0.4 (0.003)\\
 STGP & 0.74 (0.021) &  0.73 (0.012) & 0.73 (0.003) &  0.44 (0.068) &  0.32 (0.016) & 0.31 (0.005)\\
\bottomrule
\end{tabular}
\caption{Advection-diffusion inverse problem: comparing (i) posterior estimates of parameter $u_0$ in terms of relative error of mean $\mathrm{REM}=\frac{\Vert \hat u_0 - u_0^\dagger \Vert}{\Vert u_0^\dagger \Vert}$ and (ii) the forward predictions $\mG(u)(\bx,t_*)$ in terms of relative error $\frac{\Vert \bar{\mG}(\bx,t_*) - \mG(u_0^\dagger)(\bx,t_*) \Vert}{\Vert \mG(u_0^\dagger)(\bx,t_*) \Vert}$ by two likelihood models (static and STGP). Each experiment is repeated for 10 runs of MCMC (pCN, $\infty$-MALA, and $\infty$-HMC respectively) and the numbers in the bracket are standard deviations of these repeated experiments.}
\label{tab:adif_comparelik}
\end{table}

We compare two likelihood models \eqref{eq:static} and \eqref{eq:STGP}.
The static model \eqref{eq:static} is commonly used in the literature of Bayesian inverse problems \cite{LAN2019a,villa2020,lan2022}. Here the STGP model \eqref{eq:STGP} is considered to better account for the spatiotemporal relationships in the data. 
We estimate the variance parameter of the joint kernel from data. The correlation length parameters are determined ($\ell_\bx=0.5$ and $\ell_t=0.2$) by investigating their autocorrelations as in Figure \ref{fig:adif_obs_acf}.
Figure \ref{fig:adif_comparelik} compares the maximum a posterior (MAP) of the parameter $u_0$ by the two likelihood models (right two panels) with the true parameter $u_0^\dagger$ (left panel). The STGP model yields a better MAP estimate closer to the truth compared with the static model.

We also run MCMC algorithms (pCN, $\infty$-MALA, and $\infty$-HMC) to estimate $u_0$. For each algorithm, we run 6000 iterations and burn in the first 1000. The remaining 5000 samples are used to obtain the posterior estimate $\hat u_0$ (Figure \ref{fig:adif_postmean_comparelik}) and posterior standard deviation (Figure \ref{fig:adif_poststd_comparelik}). The STGP model \eqref{eq:STGP} consistently generates estimates closer to the true values (refer to Figure \ref{fig:adif_comparelik}) with smaller posterior standard deviation than the static model \eqref{eq:static} using various MCMC algorithms.
Such improvement of parameter estimation by the STGP model \eqref{eq:STGP} is also verified by smaller relative error of mean estimates $\mathrm{REM}=\frac{\Vert \hat u_0 - u_0^\dagger \Vert}{\Vert u_0^\dagger \Vert}$ reported in Table \ref{tab:adif_comparelik}, which summarizes the results of 10 repeated experiments with their standard deviations in the brackets. 


\begin{figure}[ht]
\includegraphics[width=1\textwidth,height=.3\textwidth]{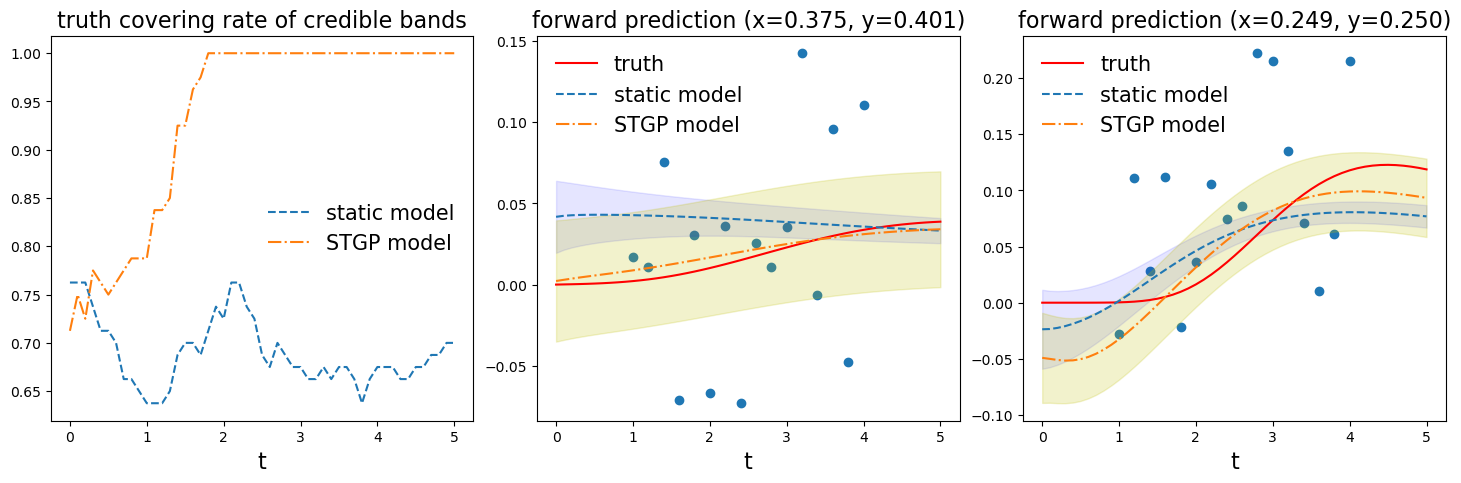}
\caption{Advection-diffusion inverse problem: comparing forward predictions, $\bar{\mG}(\bx, t_*)$, based on the static model and the STGP model. The left panel plots the curves representing the percentage of 80 (corresponding to the selected locations) credible bands that cover the true solution $\mG(u_0^\dagger)(\bx,t_*)$ at each time $t_*\in [0,5]$. The right two panels show the predicted time series (blue dashed and orange dot-dashed lines) along with the credible bands (shaded regions) by the two models compared with the truth (red solid line) at two selective locations $\bx=(0.375, 0.401)$ and $\bx=(0.249, 0.250)$. Blues dots are observations.
}
\label{fig:adif_pred_comparelik}
\end{figure}

Finally, we consider the forward prediction \eqref{eq:fwd_pred} over the time interval $[0,5]$.
We substitute each of the 5000 samples $\{u^{(s)}\}_{s=1}^{5000}$ generated by $\infty$-HMC into $\mG(u^{(s)})(\bx, t_*)$ to solve the advection-diffusion equation \eqref{eq:adif} for $t_*\in [0,5]$. We observe each of these 5000 solutions at the 80 locations (Figure \ref{fig:spatiotemporal_obs}) for $50$ points equally spaced in $[0,5]$. Then we obtain the prediction by $\bar{\mG}(\bx,t_*)_{80\times 50}=\frac{1}{5000}\sum_{s=1}^{5000}\mG(u^{(s)})(\bx, t_*)$, and compute the relative errors in terms of the Frobenius norm of the difference between the prediction and the true solution $\mG(u_0^\dagger)(\bx,t_*)$: $\frac{\Vert \bar{\mG}(\bx,t_*) - \mG(u_0^\dagger)(\bx,t_*) \Vert}{\Vert \mG(u_0^\dagger)(\bx,t_*) \Vert}$. Table \ref{tab:adif_comparelik} shows the STGP model \eqref{eq:STGP} provides more accurate predictions with smaller errors compared with the static model \eqref{eq:static}.
Figure \ref{fig:adif_pred_comparelik} depicts the predicted time series $\bar{\mG}(\bx, t_*)$ at two selective locations based on the static (blue dashed line) and the STGP (orange dot-dashed line) models along with their credible bands (shaded regions) compared with the truth (red solid lines) in the two right panels. Note that with smaller credible bands, the static model is more certain about its prediction that is further away from the truth. While the STGP model provides wider credible bands that cover more of the true trajectories, indicating a more appropriate uncertainty being quantified. Therefore, on the left panel of Figure \ref{fig:adif_pred_comparelik}, the STGP model has higher truth covering rate for its credible intervals among these $80$ locations on most of $t_*\in[0,5]$. Note these models are trained on $t\in[1,4]$, so the STGP model does not show much advantage at the beginning but quickly outperforms the static model after $t_*=1$.

\begin{figure}[t]
\includegraphics[width=.66\textwidth,height=.35\textwidth]{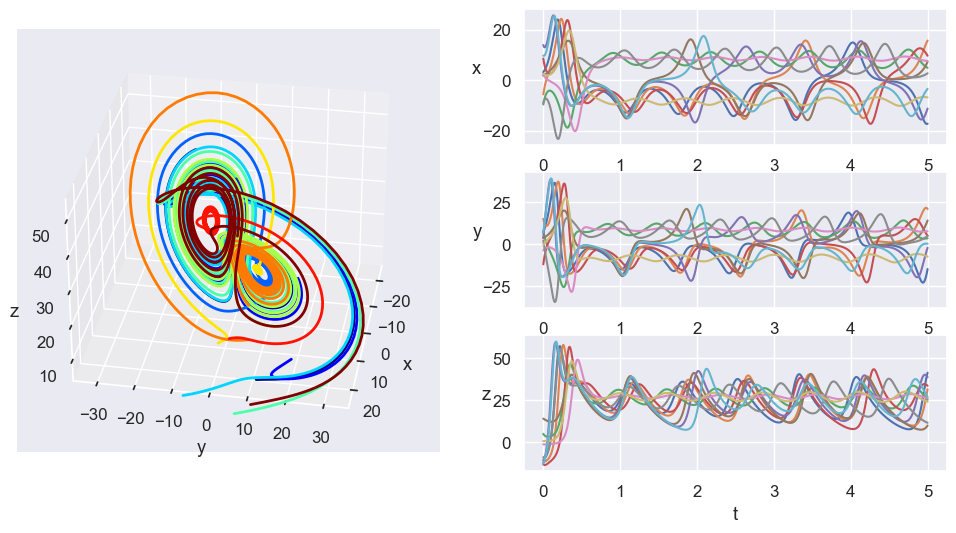}
\includegraphics[width=.33\textwidth,height=.35\textwidth]{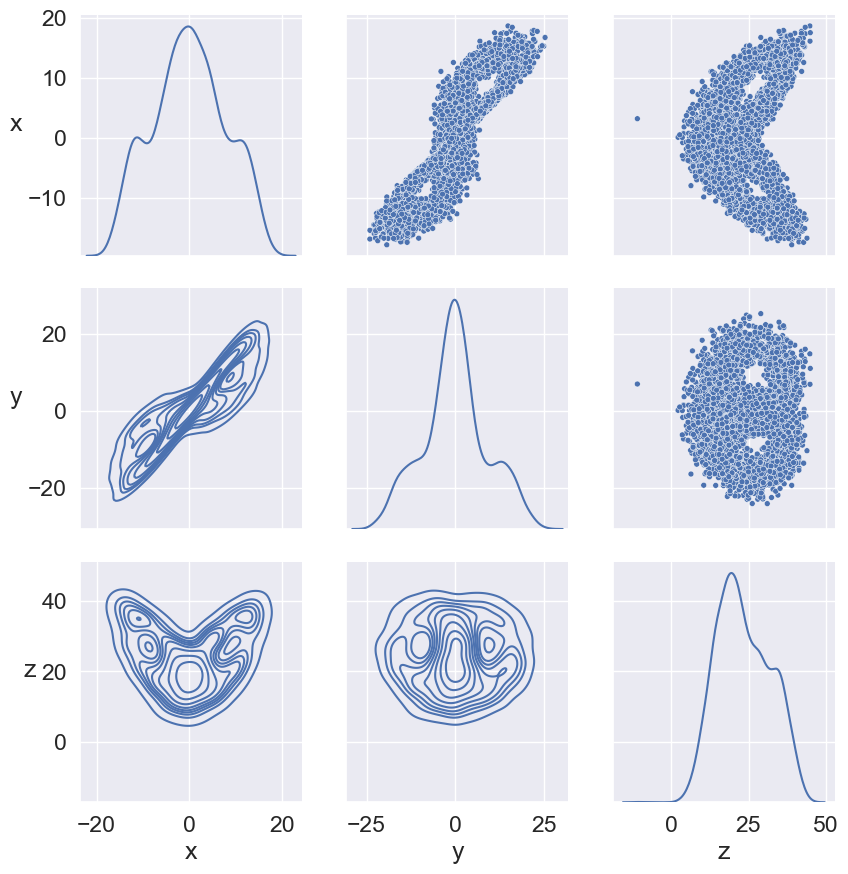}
\caption{Lorenz63 dynamics: two-lobe orbits (left), chaotic solutions (middle) and coordinates' distributions (right).}
\label{fig:lrz_chaotic}
\end{figure}

\subsection{Chaotic dynamical inverse problems}\label{sec:chaotic}
Chaos, refers to the behavior of a dynamical system that appears to be random in long term even its evolution is fully determined by the initial condition. 
Many physical systems are characterized by the presence of chaos that has been extensively demonstrated \cite{Lorenz_1963,Ivancevic_2008,Bishop_2017}. 
The main challenges of analyzing chaotic dynamical systems include the stability, the transitivity, and the sensitivity to the initial conditions (which contributes to the seeming randomness) \cite{effah2018study}. In the study of chaotic dynamical systems, one of the interests is determining the essential system parameters given the observed data.
In this section, we will investigate three chaotic dynamical systems, Lorenz63 \cite{Lorenz_1963}, R\"ossler \cite{agiza2001synchronization} and Chen \cite{yassen2003chaos}, that can be summarized as the first-order ODE: 
    $\dot\bx = f(\bx; u)$. 
We will apply the CES framework (Section \ref{sec:BUQ}) to learn the system parameter $u$ and quantify its associated uncertainty based on the observed trajectories. We find the spatiotemporal models numerically more advantageous by fitting the whole trajectories than the common approach by averaging the trajectories over time \cite{Schneider2017,cleary2020,Huang_2022}.

\subsubsection{Lorenz system}\label{sec:Lorenz}
The most popular example of chaotic dynamics is the Lorenz63 system \cite[named after the author and the year it was proposed in][]{Lorenz_1963} that represents a simplified model of atmospheric convection  for the chaotic behavior of the weather. 
The governing equations of the Lorenz system are given by the following ODE
\begin{equation}\label{eq:Lorenz}
\begin{cases}
    \dot{x} &= \sigma (y-x),\\
    \dot{y} &= x(\rho -z)-y,\\
    \dot{z} &= x y -\beta z, 
\end{cases}
\end{equation}
where $x$, $y$, and $z$ denote variables proportional to convective intensity, horizontal and vertical temperature differences and $u:=(\sigma, \rho, \beta)$ represents the model parameters known as Prandtl number ($\sigma$), Rayleigh number ($\rho$), and an unnamed parameter ($\beta$) used for physical proportions of the regions \cite{ott1981strange}.

\begin{figure}[t]
\includegraphics[width=.49\textwidth,height=.4\textwidth]{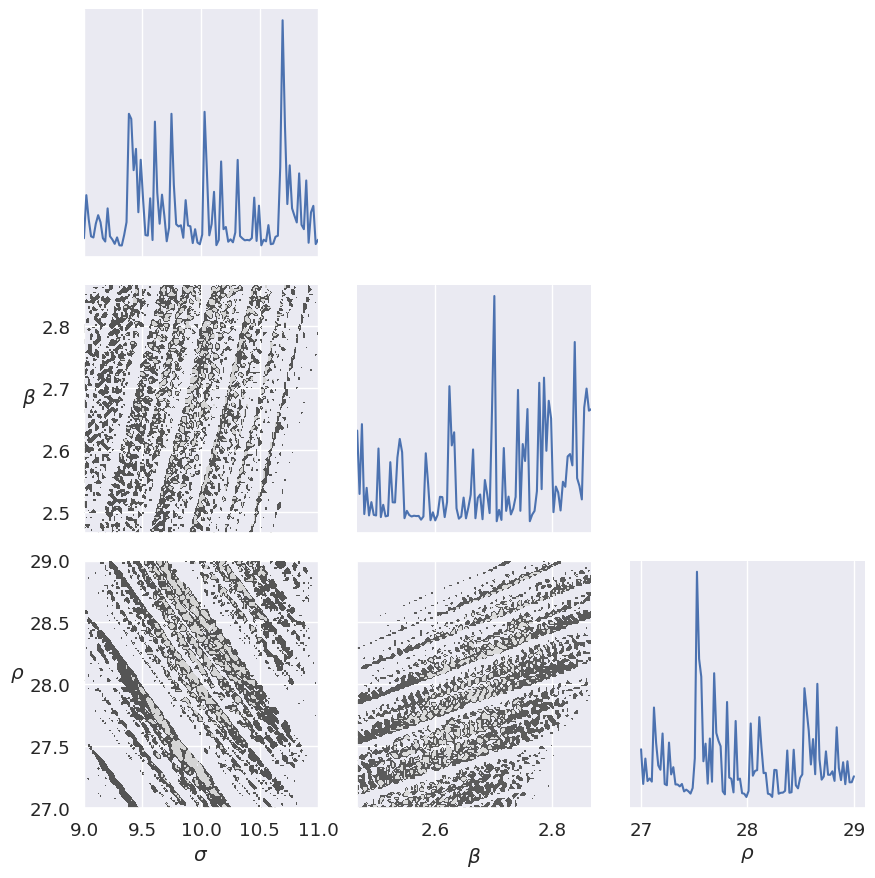}
\includegraphics[width=.49\textwidth,height=.4\textwidth]{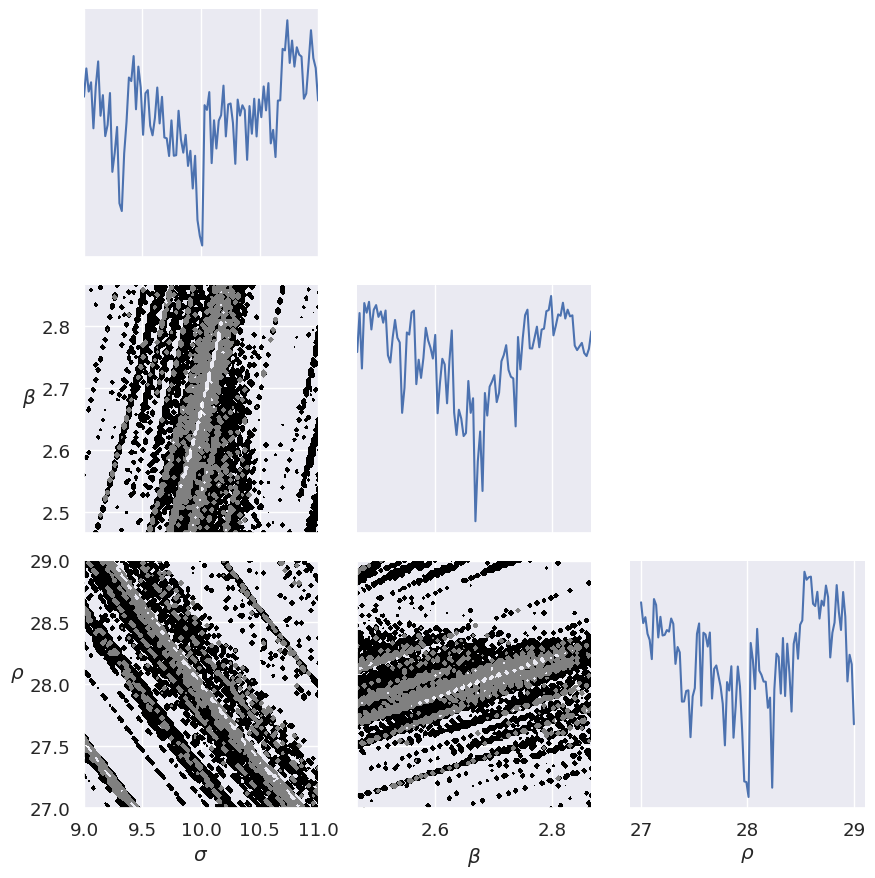}
\caption{Lorenz inverse problem: marginal (diagonal) and pairwise (lower triangle) sections of the joint density $p(u)$ by the time-averaged model (left) and the STGP model (right) respectively.}
\label{fig:lrz_pairpdf}
\end{figure}

The behavior of Lorenz63 system \eqref{eq:Lorenz} strongly relies on the parameters. In many studies, the parameter $\rho$ varies in $(0, \infty)$ and the other parameters $\sigma$ and $\beta$ are held constant. 
In particular, \eqref{eq:Lorenz} has a stable equilibrium point at the origin for $\rho \in (0, 1)$. For $\rho \in (1, \gamma)$ with $\gamma = \sigma\frac{\sigma + \beta + 3}{\sigma - \beta -1}$, \eqref{eq:Lorenz} has three equilibrium points, one unstable equilibrium point at the origin and two stable equilibrium points at $\tp{(\sqrt{\beta(\rho-1)}, \sqrt{\beta(\rho-1)}, \rho -1)}$ and  $\tp{(-\sqrt{\beta(\rho-1)}, -\sqrt{\beta(\rho-1)}, \rho -1)}$. When $\rho > \gamma$, the equilibrium points become unstable and it results in unstable spiral shaped trajectories. 
One classical configuration for the parameter in \eqref{eq:Lorenz} is $\sigma = 10$, $\beta = \frac{8}{3}$, $\rho = 28$ 
when the system exhibits two-lobe orbits, also known as the butterfly effect \cite{yang2002control} (See the left panel of Figure \ref{fig:lrz_chaotic}).
In this example, we seek to infer such parameter $u^\dagger=(\sigma^\dagger, \beta^\dagger, \rho^\dagger)= (10, 8/3, 28)$ based on the observed chaotic trajectories demonstrated in the middle panel of Figure \ref{fig:lrz_chaotic}.
Note the solutions $(x(t), y(t), z(t))$ highly depend on the initial conditions $(x(0), y(0), z(0))$, we hence fix $(x(0), y(0), z(0))$ in the following.

Due to the chaotic nature of the states $\{(x(t), y(t), z(t)) : t\in [0, \tau]\}$, we can treat these coordinates as random variables. In right panel of Figure \ref{fig:lrz_chaotic}, we demonstrate their marginal and pairwise distributions (diagonal and lower triangle) estimated by a collection of states (upper triangle) along a long-time trajectory solved with $u^\dagger$.
For a given parameter $u=(\sigma, \beta, \rho)$, we have the trajectory $\mG(u)$ as the following map:
\begin{equation}\label{eq:traj}
\mG(u): \; \mbR_+ \rightarrow \mbR^3, \quad t \mapsto (x(t;u), y(t;u), z(t;u))
\end{equation}
where $(x(t;u), y(t;u), z(t;u))$ is the solution of \eqref{eq:Lorenz} for given parameter $u$.
We generate spatiotemporal data from the chaotic dynamics \eqref{eq:Lorenz} with $u^\dagger=(\sigma^\dagger, \beta^\dagger, \rho^\dagger)$ by observing its trajectory on $J=100$ equally spaced time points $t_j\in[t_0, t_0+T]$: $\bX(u^\dagger)_{3\times 100}:=\{\mG(u^\dagger)(t_j)=(x(t_j;u^\dagger), y(t_j;u^\dagger), z(t_j;u^\dagger))\}_{j=1}^J$. These observations can be viewed as a 3-dimensional time series that estimate the empirical covariance $\Gamma_\textrm{obs}$ as in \cite{cleary2020}.
The inverse problem involves learning the parameter $u$ given these observations, also known as parameter identification \cite{Negrini_2021}.

\begin{figure}[t]
\includegraphics[width=1\textwidth,height=.5\textwidth]{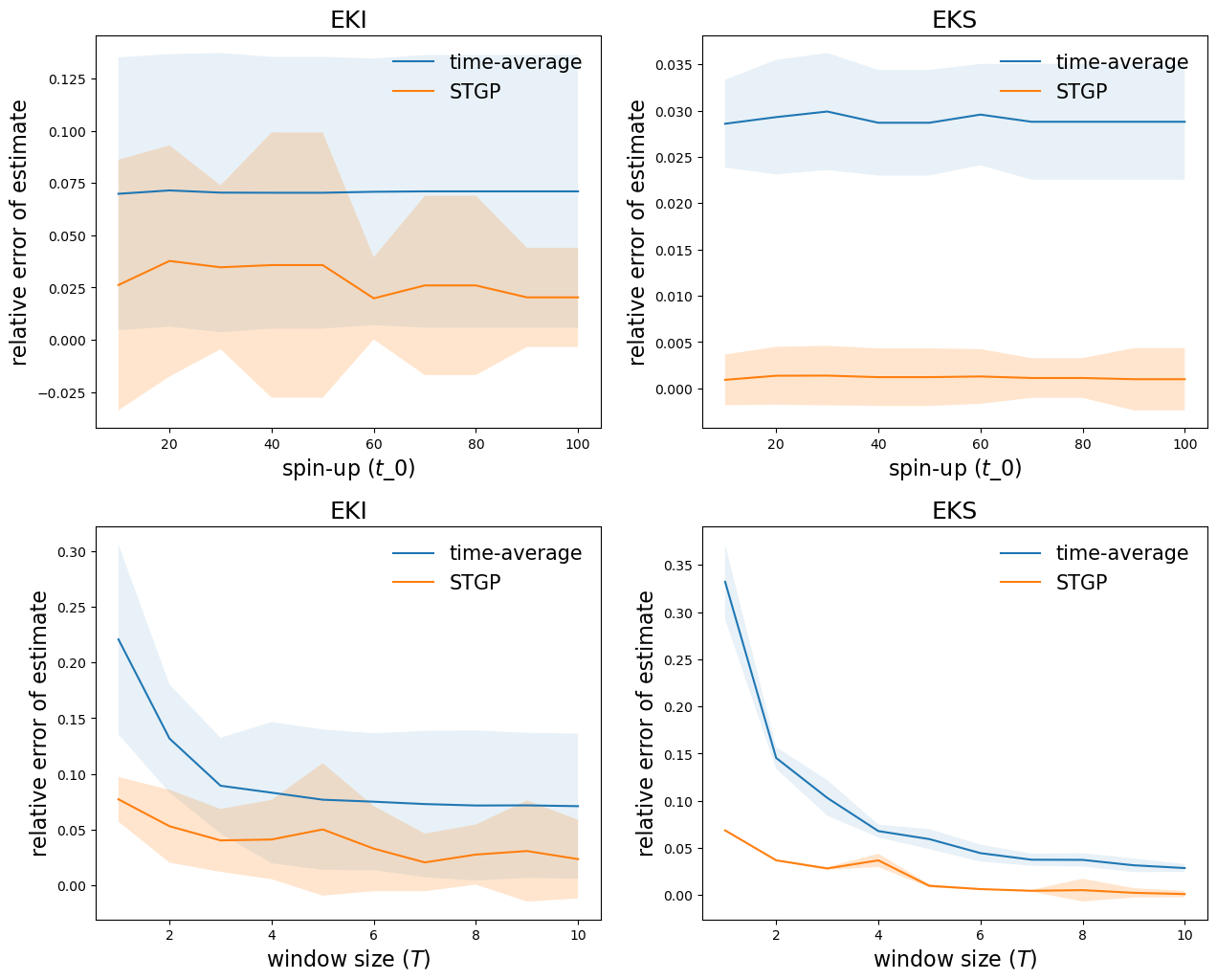}
\caption{Lorenz inverse problem: comparing posterior estimates of parameter $u$ for two models (time-average and STGP) in terms of relative error of mean $\mathrm{REM}=\frac{\Vert \hat u - u^\dagger \Vert}{\Vert u^\dagger \Vert}$. The upper row shows the results by varying the spin-up $t_0$ and fixing $T=10$. The lower row shows the results by varying the observation window size $T$ and fixing $t_0=100$. Each experiment is repeated for 10 runs of EnK (EKI and EKS respectively) with $J=500$ ensembles and the shaded regions indicate standard deviations of such repeated experiments.}
\label{fig:lrz_rem_spinavgs}
\end{figure}

Following \cite{cleary2020}, we endow a log-Normal prior on $u$: $\log u \sim \mN(\mu_0, \sigma_0^2)$ with $\mu_0=(2.0, 1.2, 3.3)$ and $\sigma_0=(0.2, 0.5, 0.15)$. 
We compare the two likelihood models \eqref{eq:time-average} and \eqref{eq:STGP} for this dynamical inverse problem.
For the time-averaged model \eqref{eq:time-average}, instead of the 3-dimensional time series from the trajectory \eqref{eq:traj}, we substitute $\bar{\bX}(u)_{3\times 1}$ with $\bar{\bX^\star}(u)_{9\times 1}=\mO \mG^\star(u)(t)$ by averaging the following augmented trajectory $\mG^\star(u)(t)$ in time \cite{cleary2020}:
\begin{equation*}
\mG^\star(u)(t) = (x(t), y(t), z(t), x^2(t), y^2(t), z^2(t), x(t)y(t), x(z)z(t), y(t)z(t))
\end{equation*}
For the spatiotemporal likelihood model STGP \eqref{eq:STGP}, we set the correlation length $\ell_\bx=0.4$ and $\ell_t=0.1$ for the spatial kernel $\mC_\bx$ and the temporal kernel $\mC_t$ respectively. They are chosen to reflect the spatial and temporal resolutions.

We first notice that the spatiotemporal modeling facilitates the learning of the true parameter $u^\dagger$.
As illustrated in Figure \ref{fig:lrz_pairpdf}, despite of the rough landscape, the marginal (e.g. $p(\sigma, \beta^\dagger, \rho^\dagger)$) and pairwise (e.g. $p(\sigma, \beta, \rho^\dagger)$) sections of the joint density $p(u)$ by the STGP model \eqref{eq:STGP} are more convex in the neighbourhood of $u^\dagger$ compared with the time-averaged model \eqref{eq:time-average}. 
This verifies the implication of Theorem \ref{thm:convexity2} on their difference in convexity.
Therefore, particle based algorithms such as EnK methods have higher chance of concentrating their ensemble particles around the true parameter value $u^\dagger$, leading to better estimates.
Here, the roughness of the posterior creates barrier for the direct application of MCMC algorithms. Therefore, we apply more robust EnK methods for parameter estimation.

We run each of the EnK algorithms for $N=50$ iterations and choose the ensembles (of size $J$) when its ensemble mean attains the minimal error in estimating the parameter $u$ with reference to its true value $u^\dagger$.
In practice, EnK algorithms usually converge quickly within a few iterations so $N=50$ suffices the need for most applications.

To investigate the roles of spin-up length $t_0$ and observation window size $T$, we run EnK multiple times while varying each of the two quantities one at a time. Seen from Figure \ref{fig:lrz_rem_spinavgs}, we observe consistently smaller errors by the STGP model \eqref{eq:STGP} compared with the time-averaged model \eqref{eq:time-average}. More specifically, the upper row indicates that the estimation errors, measured by $\mathrm{REM}=\frac{\Vert \hat u - u^\dagger \Vert}{\Vert u^\dagger \Vert}$, are not every sensitive to the spin-up $t_0$ given sufficient window size $T=10$. On the other hand, for fixed spin-up $t_0=100$, both models decrease errors with increasing window size $T$ as they aggregate more information. However, the STGP model requires only about $\frac{1}{4}$ time length as the time-averaged model to attain accuracy at the same level ($T=1$ vs $T=4$). This supports that the STGP is preferable to the time-average approach as the former may add a small overhead for the statistical inference but could save much more in resolving the physics (solving ODE/PDE), which is usually more expensive.

\begin{figure}[t]
\includegraphics[width=.49\textwidth,height=.4\textwidth]{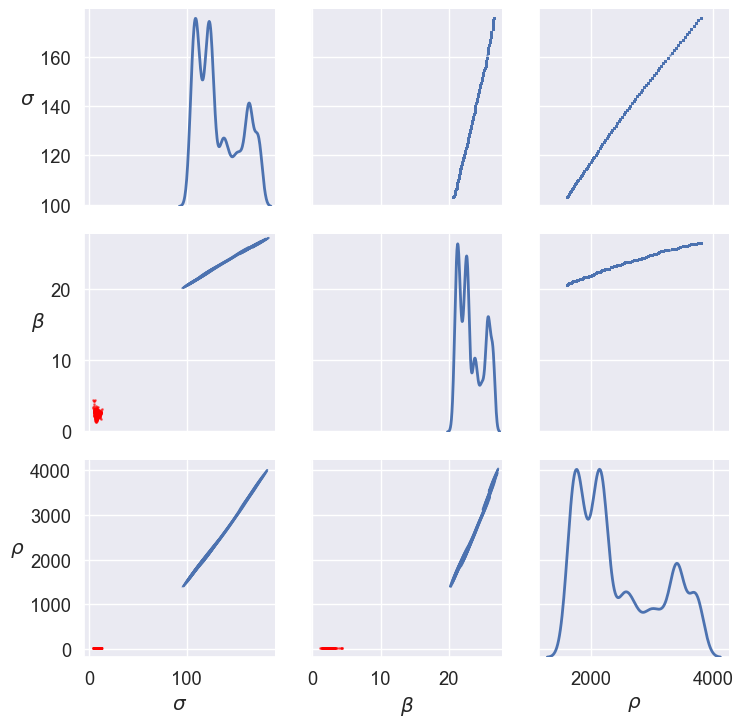}
\includegraphics[width=.49\textwidth,height=.4\textwidth]{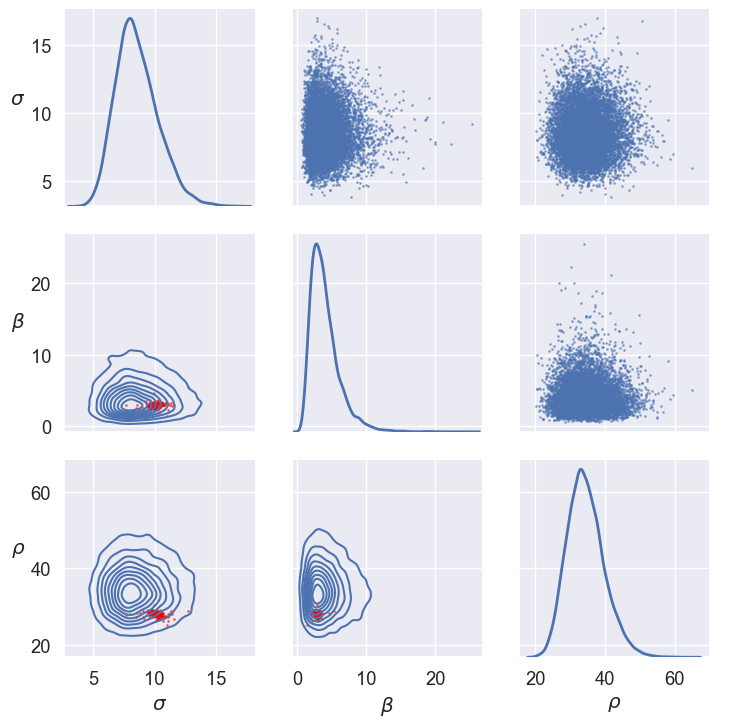}
\caption{Lorenz inverse problem: marginal (diagonal) and pairwise (lower triangle) distributions estimated with 10000 samples (upper triangle) by the pCN algorithm based on NN emulators for the time-averaged model (left) and the STGP model (right) respectively. Red dots (lower triangle) are selective 10000 ensemble particles from running the EKS algorithm.}
\label{fig:lrz_UQ}
\end{figure}

\begin{table}[ht]\scriptsize
\centering
\begin{tabular}{l|lllll}
\toprule
Model-Algorithms &        J=50 &           J=100 &           J=200 &               J=500 &              J=1000 \\
\midrule
      Tavg-EKI & 0.06 (0.03) &     0.09 (0.03) &     0.09 (0.01) &         0.06 (0.04) &         0.07 (0.02) \\
      Tavg-EKS & 0.10 (0.02) & 0.07 (4.62e-03) & 0.05 (2.60e-03) &     0.03 (3.04e-03) &     0.03 (8.56e-04) \\
       STGP-EKI & 0.07 (0.03) &     0.04 (0.03) &     0.03 (0.02) &         0.02 (0.03) &         0.02 (0.01) \\
       STGP-EKS & 0.09 (0.03) &     0.05 (0.03) &     0.03 (0.02) & 3.97e-04 (1.06e-03) & 5.52e-04 (6.37e-04) \\
\bottomrule
\end{tabular}
\caption{Lorenz inverse problem: comparing posterior estimates of parameter $u$ for two models, time-average (Tavg) and STGP, in terms of relative error of median $\mathrm{REM}=\frac{\Vert \hat u - u^\dagger \Vert}{\Vert u^\dagger \Vert}$. Each experiment is repeated for 10 runs of EnK (EKI and EKS respectively) and the numbers in the bracket are standard deviations of such repeated experiments.}
\label{tab:lrz_rem}
\end{table}

Now we set spin-up $t_0=100$ long enough to ignore effect of initial condition in the dynamics and choose the observation window size $T=10$.
We compare the two models \eqref{eq:time-average} \eqref{eq:STGP} using EnK algorithms with different ensemble sizes ($J$) to obtain an estimate $\hat u$ of the parameter $u$. Figure \ref{fig:lrz_rem} shows that the STGP model performs better than the time-averaged model in generating smaller errors (REM) for almost all cases. In general more ensembles help reduce the errors except for the time-averaged model using EKI algorithm. Note, the STGP model with EKS algorithm yields parameter estimates with the lowest errors.
Table \ref{tab:lrz_rem} summarizes the REM's by different combinations of the two likelihood models (time-averaged and STGP) and two EnK algorithms (EKI and EKS). Again we can see consistent advantage of the spatiotemporal likelihood model STGP \eqref{eq:STGP} over the simple time-averaged model \eqref{eq:time-average} in producing more accurate parameter estimation.

\begin{figure}[t]
\includegraphics[width=1\textwidth,height=.3\textwidth]{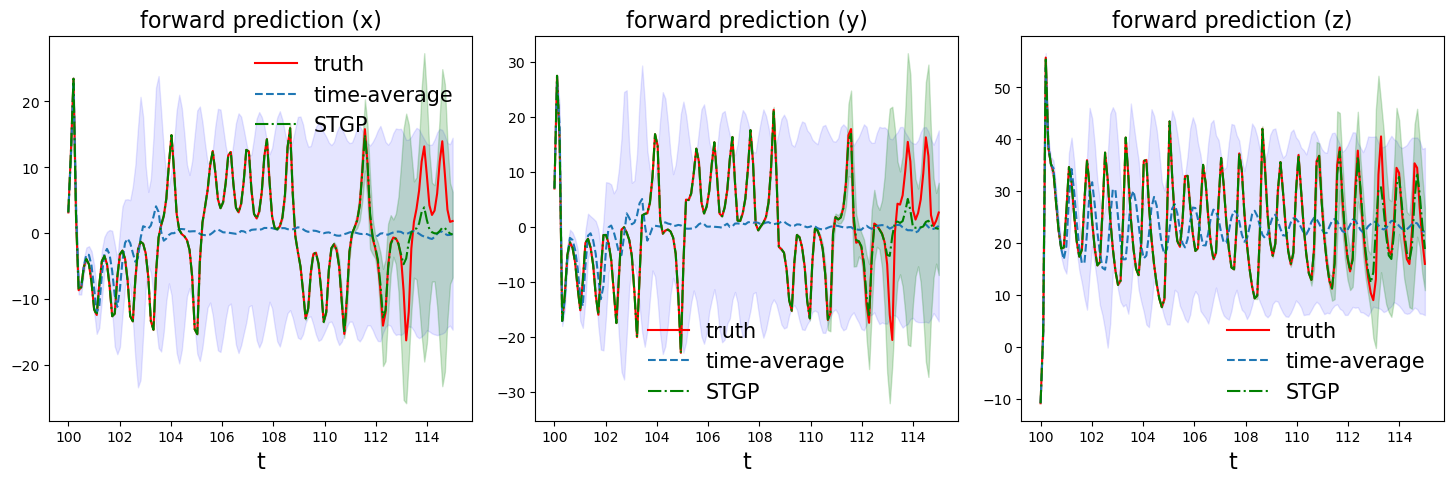}
\caption{Lorenz inverse problem: comparing forward predictions $\bar{\mG}(\bx, t_*)$ based on the time-averaged model and the STGP model.}
\label{fig:lrz_pred_comparelik}
\end{figure}

Next, we apply CES (Section \ref{sec:BUQ}) \cite{cleary2020,lan2022} to quantify the uncertainty of the estimate $\hat u$. Direct application of MCMC suffers from the extremely low acceptance rate because of the rough density landscape (Figure \ref{fig:lrz_pairpdf}). Ensemble particles from EnK algorithm cannot provide rigorous systematic UQ due to the ensemble collapse \cite{Schillings_2017a,Schillings_2017b,deWiljes2018,chada2019} (See red dots in Figure \ref{fig:lrz_UQ}). Therefore, we run approximate MCMC based on NN emulators built from EnK outputs $\{u_n^{(j)}, \mG(u_n^{(j)})\}_{j=1,n=0}^{J,N}$.
Note, we have different structures for the observed data in the two models \eqref{eq:time-average} \eqref{eq:STGP}: 9-dimensional summary of time series for the time-averaged model \eqref{eq:time-average} and $3\times 100$ time series for the STGP model \eqref{eq:STGP}. Therefore we build densely connected NN (DNN) $\mG^e: \mbR^3\to\mbR^9$ for the former and DNN-RNN (recurrent NN) type of network $\mG^e: \mbR^3\to\mbR^{3\times 100}$ for the latter to account for their different data structures in the forward output.
Figure \ref{fig:lrz_UQ} compares the marginal (diagonal) and pairwise (lower triangle) posterior densities of $u$ estimated by 10000 samples (upper triangle) of the pCN algorithm based on the corresponding NN emulators for the two models. The spatiotemporal model STGP \eqref{eq:STGP} yields more reasonable UQ results compared with the time-averaged model \eqref{eq:time-average}.

Finally, we consider the forward prediction $\bar{\mG}(\bx, t_*)$ \eqref{eq:fwd_pred} for $t_*\in[t_0, t_0+1.5T]$ with $J=500$ EKS ensembles corresponding to the lowest error.
Figure \ref{fig:lrz_pred_comparelik} compares the prediction results given by these two models. The result by the STGP model is very close to the truth till $t=113$ while the prediction by the time-averaged model quickly departs from the truth only after $t=102$. The STGP model predicts the future of this challenging chaotic dynamics significantly better than the time-averaged model.

\subsubsection{R\"ossler system}\label{sec:Rossler}
\begin{figure}[t]
\includegraphics[width=.6\textwidth,height=.35\textwidth]{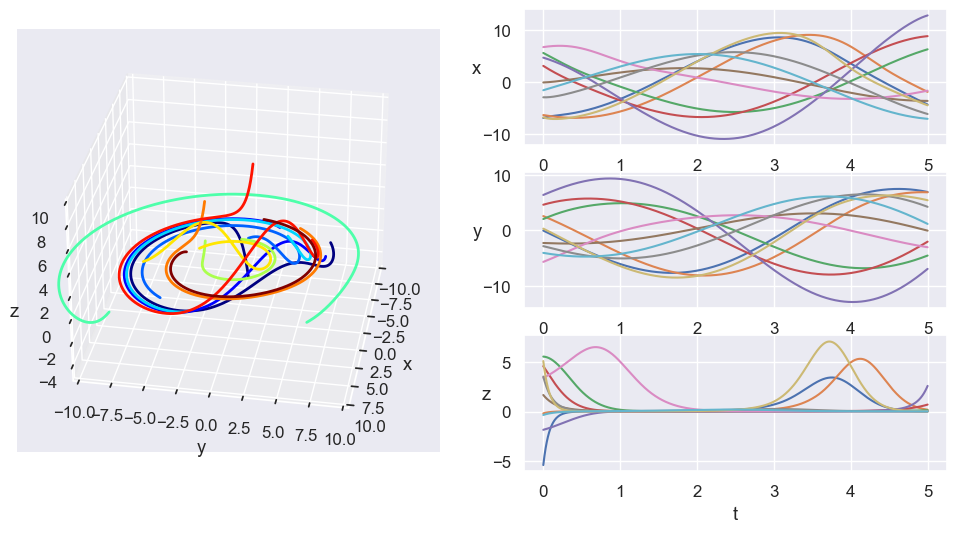}
\includegraphics[width=.33\textwidth,height=.35\textwidth]{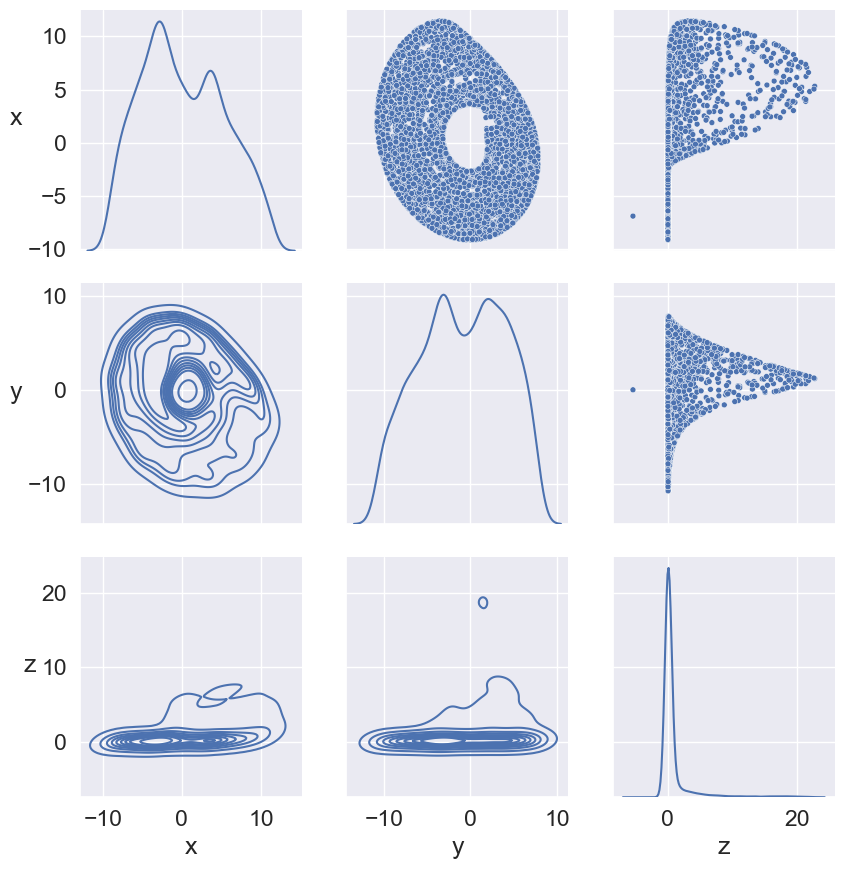}
\caption{R\"ossler dynamics: single-lobe orbits (left), chaotic solutions (middle) and coordinates' distributions (right).}
\label{fig:rsl_chaotic}
\end{figure}


Next we consider the following R\"ossler system \cite{hegazi2001controlling} governed by the system of autonomous differential equations:
\begin{equation}\label{eq:Rossler}
\begin{cases}
    \dot{x} &= -y - z, \\
    \dot{y} &= x+a y, \\
    \dot{z} &= b + z(x-c). 
\end{cases}
\end{equation}
where $a, b, c>0$ are parameters determining the behavior of the system.
The R\"ossler attractor was originally discovered by German biochemist Otto Eberhard Rössler \cite{ROSSLER_1976,ROSSLER_1979}. 
When $c^2 > 4ab$,
the system \eqref{eq:Rossler} exhibits continuous-time chaos and has two unstable equilibrium points $(a\gamma_{-}, -\gamma_{-}, \gamma_{-})$ and $(a\gamma_{+}, -\gamma_{+}, \gamma_{+})$ with $\gamma_{+} = \frac{c + \sqrt{c^2-4ab}}{2a},\, \gamma_{-} = \frac{c - \sqrt{c^2-4ab}}{2a}$.
Note that the R\"ossler attractor has similarities to the Lorenz attractor, nevertheless it has a single lobe and offers more flexibility in the qualitative analysis. The true parameter that we try to infer is $u^\dagger=(a^\dagger, b^\dagger, c^\dagger)=(0.2, 0.2, 5.7)$. 
Figure \ref{fig:rsl_chaotic} illustrates the single-lobe orbits (left), the chaotic solutions (middle) and their marginal and pairwise distributions (right) of their coordinates viewed as random variables.

Note, the R\"ossler dynamics evolve at a lower rate compared with the Lorenz63 dynamics (compare the middle panels of Figure \ref{fig:rsl_chaotic} and Figure \ref{fig:lrz_chaotic}). Therefore, we adopt a longer spin-up length ($t_0=1000$) and a larger window size ($T=100$).
For the STGP model \eqref{eq:STGP}, spatiotemporal data are generated by observing the trajectory \eqref{eq:traj} of the chaotic dynamics \eqref{eq:Rossler} with $u^\dagger=(a^\dagger, b^\dagger, c^\dagger)$ for $J=100$ time points in $[t_0,t_0+T]$.
We also augment the time-averaged data with second-order moments for the time-averaged model \eqref{eq:time-average}.
In this Bayesian inverse problem, we adopt a log-Normal prior on $u$: $\log u \sim \mN(\mu_0, \sigma_0^2)$ with $\mu_0=(-1.5, -1.5, 2.0)$ and $\sigma_0=(0.15, 0.15, 0.2)$. 
Once again, with spatiotemporal likelihood model STGP \eqref{eq:STGP}, learning the true parameter value $u^\dagger$ becomes easier because the posterior density $p(u)$ concentrates more on $u^\dagger$ compared with the time-averaged model \eqref{eq:time-average}, as indicated by Theorem \ref{thm:convexity2}. See Figure \ref{fig:rsl_pairpdf} for the comparison on their marginal and pairwise sections of the joint density $p(u)$.

\begin{figure}[t]
\includegraphics[width=1\textwidth,height=.5\textwidth]{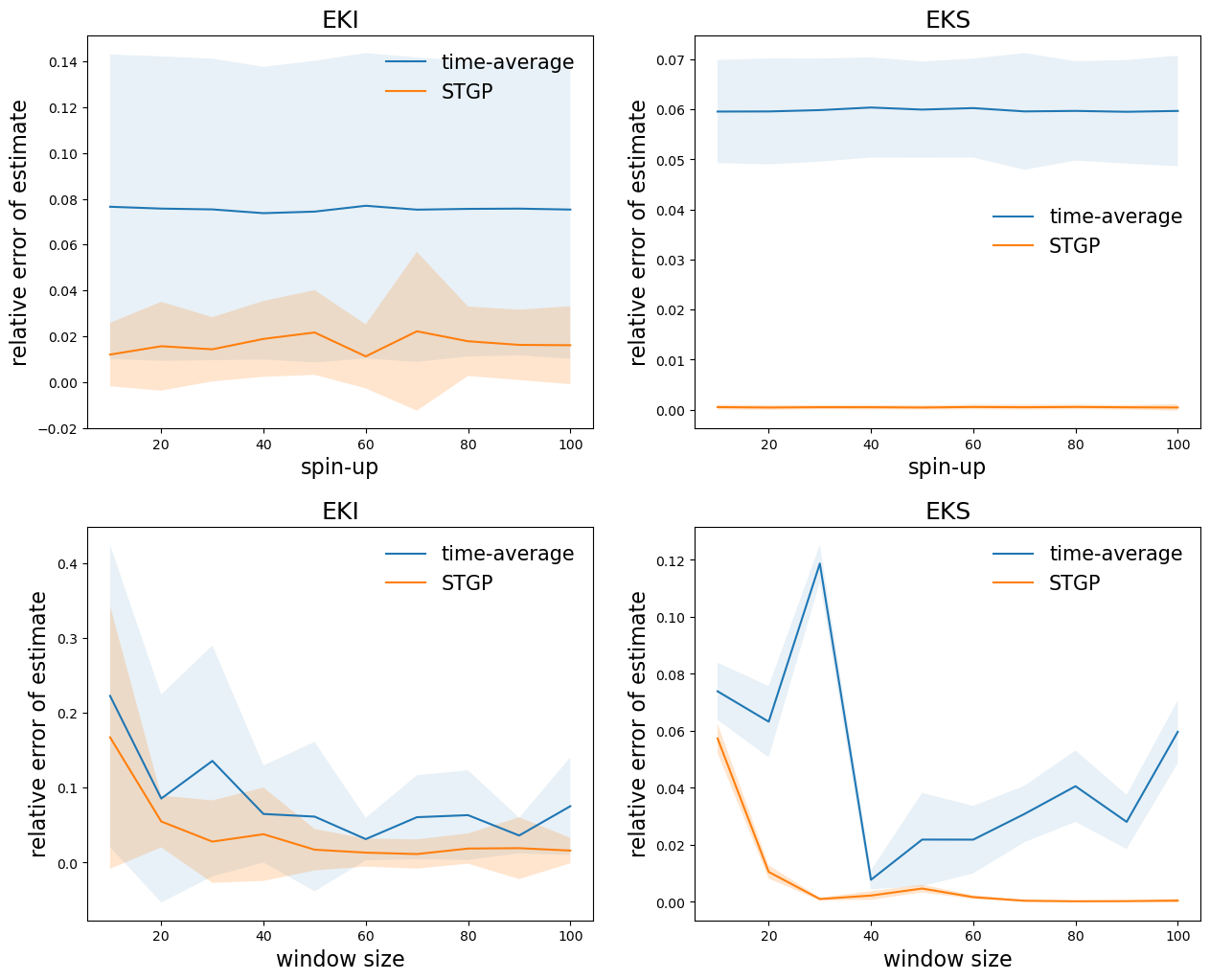}
\caption{R\"ossler inverse problem: comparing posterior estimates of parameter $u$ for two models (time-average and STGP) in terms of relative error of mean $\mathrm{REM}=\frac{\Vert \hat u - u^\dagger \Vert}{\Vert u^\dagger \Vert}$. The upper row shows the results by varying the spin-up $t_0$ and fixing $T=100$. The lower row shows the results by varying the window size $T$ and fixing $t_0=100$. Each experiment is repeated for 10 runs of EnK (EKI and EKS respectively) with $J=500$ ensembles and the shaded regions indicate standard deviations of such repeated experiments.}
\label{fig:rsl_rem_spinavgs}
\end{figure}

We also compare the two models \eqref{eq:time-average} \eqref{eq:STGP} when investigating the roles of spin-up length $t_0$ and observation window size $T$ in Figure \ref{fig:rsl_rem_spinavgs}. Despite of the consistent smaller errors (expressed in terms of REM) by the STGP model, REM is not every sensitive to the spin-up $t_0$ given sufficient window size $T=100$.
However, for fixed spin-up $t_0=100$, the STGP model \eqref{eq:STGP} is superior than the time-averaged approach \eqref{eq:time-average} in reducing the estimation error using smaller observation time window $T$: the former requires only half time length as the latter to attain the same level of accuracy ($T=30$ vs $T=60$ with EKI and $T=20$ vs $T=40$ with EKS).

\begin{figure}[t]
\includegraphics[width=.49\textwidth,height=.4\textwidth]{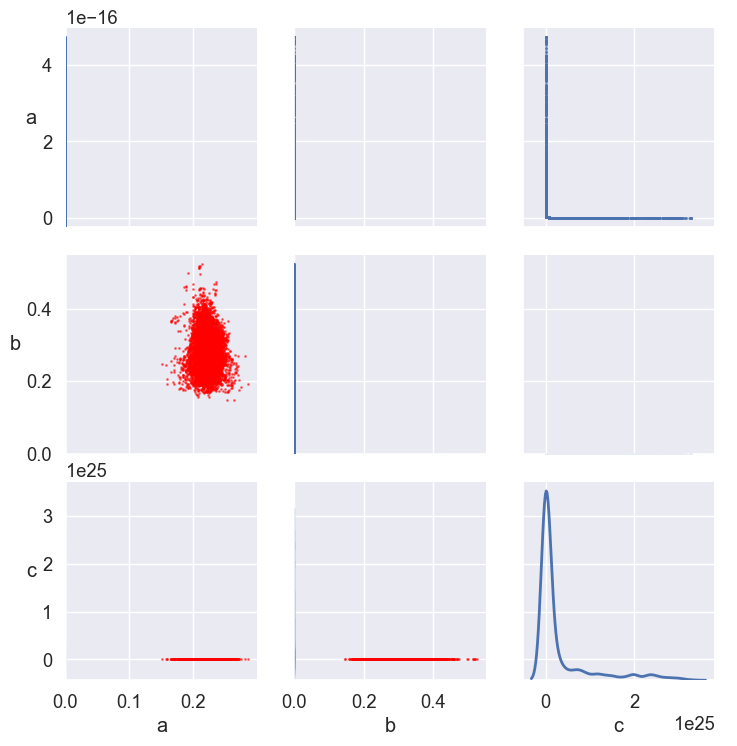}
\includegraphics[width=.49\textwidth,height=.4\textwidth]{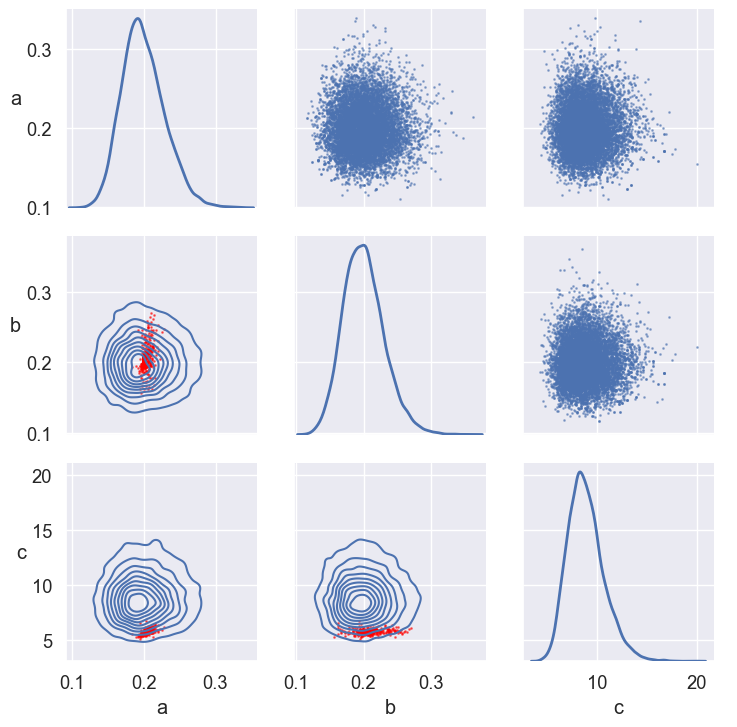}
\caption{R\"ossler inverse problem: marginal (diagonal) and pairwise (lower triangle) distributions estimated with 10000 samples (upper triangle) by the pCN algorithm based on NN emulators for the time-averaged model (left) and the STGP model (right) respectively. Red dots (lower triangle) are selective 10000 ensemble particles from running the EKS algorithm.}
\label{fig:rsl_UQ}
\end{figure}

\begin{table}[ht]\scriptsize
\centering
\begin{tabular}{l|lllll}
\toprule
Model-Algo &            J=50 &           J=100 &               J=200 &               J=500 &              J=1000 \\
\midrule
      Tavg-EKI & 0.16 (0.09) &     0.11 (0.06) &         0.10 (0.07) &         0.07 (0.04) &         0.11 (0.07) \\
      Tavg-EKS & 0.06 (0.02) & 0.06 (7.61e-03) &     0.06 (6.20e-03) &     0.06 (5.37e-03) &     0.06 (2.53e-03) \\
       STGP-EKI & 0.02 (0.02) &     0.01 (0.01) &         0.02 (0.02) &     0.01 (9.09e-03) &         0.01 (0.02) \\
       STGP-EKS & 0.02 (0.01) & 2.47e-03 (0.02) & 7.63e-04 (2.86e-03) & 4.23e-04 (2.45e-04) & 3.62e-04 (1.19e-04) \\
\bottomrule
\end{tabular}
\caption{R\"ossler inverse problem: comparing posterior estimates of parameter $u$ for two models (time-average and STGP) in terms of relative error of median $\mathrm{REM}=\frac{\Vert \hat u - u^\dagger \Vert}{\Vert u^\dagger \Vert}$. Each experiment is repeated for 10 runs of EnK (EKI and EKS respectively) and the numbers in the bracket are standard deviations of such repeated experiments.}
\label{tab:rsl_rem}
\end{table}

Now we fix $t_0=1000$ and $T=100$. 
Figure \ref{fig:rsl_rem} compares these two models \eqref{eq:time-average} \eqref{eq:STGP} in terms of REM's of the parameter estimation by EnK algorithms with different ensemble sizes ($J$).
The STGP model \eqref{eq:STGP} shows universal advantage over the time-averaged model \eqref{eq:time-average} in generating smaller REM's.
Note, the time-averaged model becomes over-fitting if running EKS more than 10 iterations, a phenomenon also reported in \cite{Iglesias_2013,Iglesias_2016}.
Table \ref{tab:rsl_rem} summarizes the REM's by different combinations of likelihood models and EnK algorithms and confirms the consistent advantage of the STGP model over the time-averaged model in rendering more accurate parameter estimation.

We apply CES (Section \ref{sec:BUQ}) \cite{cleary2020,lan2022} for the UQ. 
Based on the EKS ($J=500$) outputs, we build DNN $\mG^e: \mbR^3\to\mbR^9$ for the time-averaged model \eqref{eq:time-average} and DNN-RNN $\mG^e: \mbR^3\to\mbR^{3\times 100}$ for the STGP model \eqref{eq:STGP} to account for their different data structures.
Figure \ref{fig:rsl_UQ} compares the marginal and pairwise posterior densities of $u$ estimated by 10000 samples of the pCN algorithm based on the corresponding NN emulators for the two models. The STGP model \eqref{eq:STGP} generates more appropriate UQ results than the time-averaged model \eqref{eq:time-average} does.
Finally, we consider the forward prediction $\bar{\mG}(\bx, t_*)$ \eqref{eq:fwd_pred} for $t_*\in[t_0, t_0+1.5T]$ with $J=500$ EKS ensembles corresponding to the lowest error.
Figure \ref{fig:rsl_pred_comparelik} shows that the STGP model provides better prediction consistent with the truth throughout the whole time window while the result by the time-averaged model deviates from the truth quickly after $t=1020$.

\begin{figure}[t]
\includegraphics[width=1\textwidth,height=.3\textwidth]{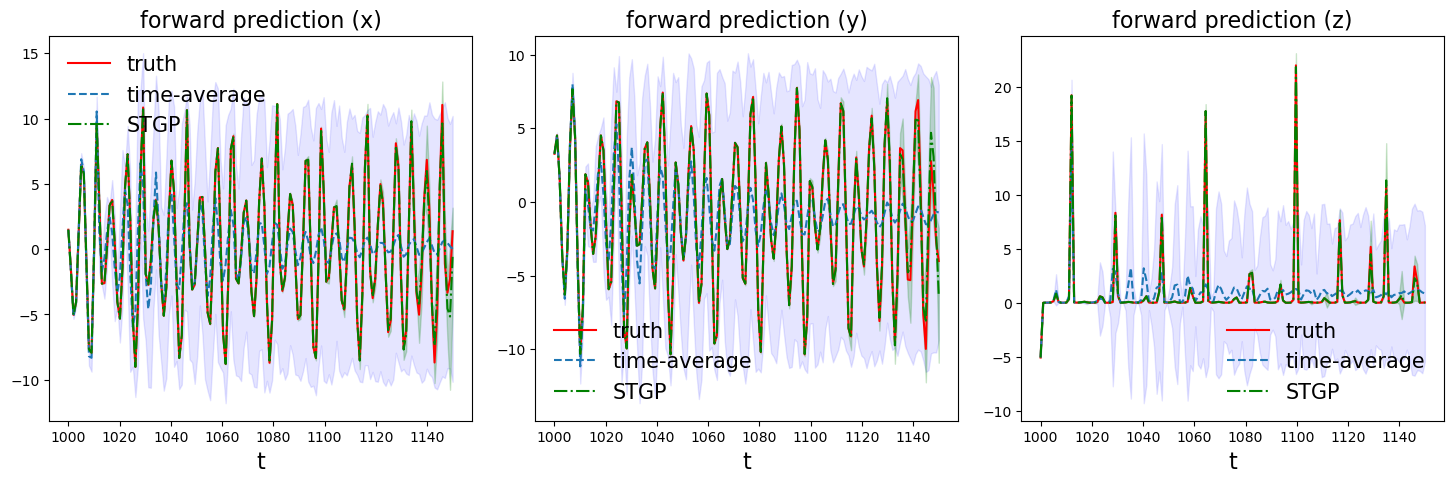}
\caption{R\"ossler inverse problem: comparing forward predictions $\bar{\mG}(\bx, t_*)$ based on the time-averaged model and the STGP model.}
\label{fig:rsl_pred_comparelik}
\end{figure}

\subsubsection{Chen system}\label{sec:Chen}
\begin{figure}[t]
\includegraphics[width=.66\textwidth,height=.35\textwidth]{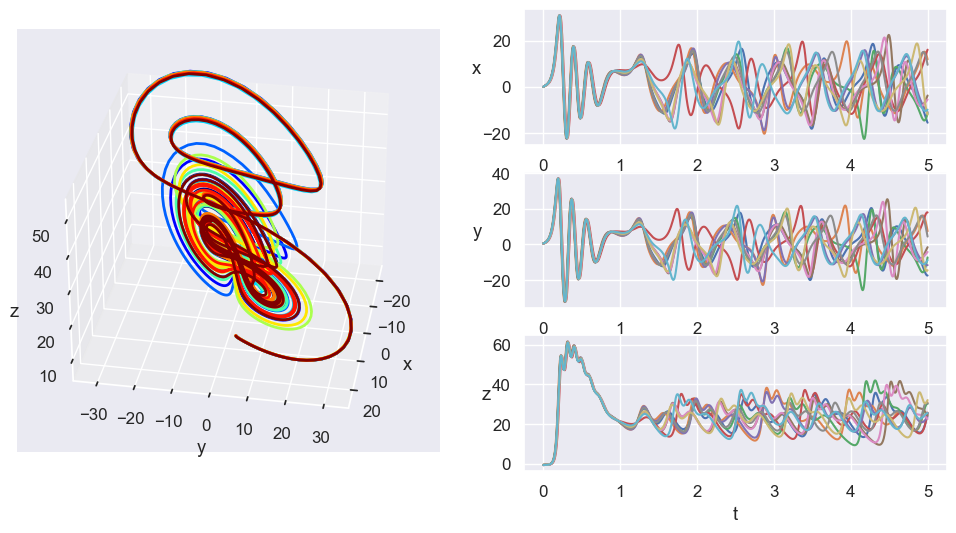}
\includegraphics[width=.33\textwidth,height=.35\textwidth]{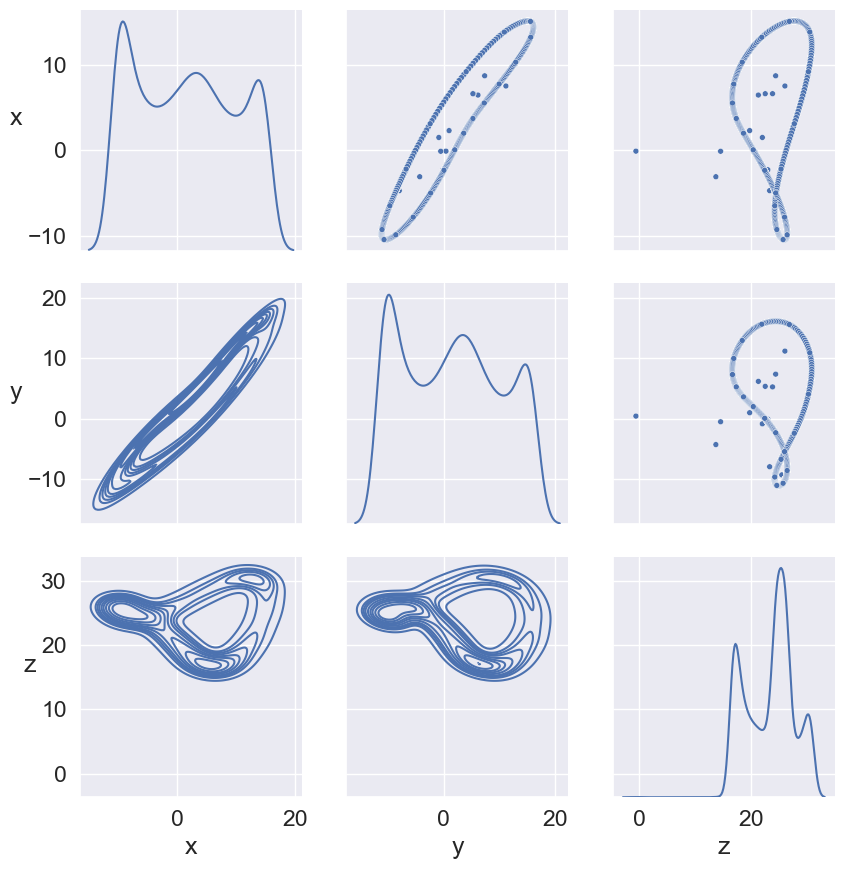}
\caption{Chen dynamics: double-scroll attractor (left), chaotic solutions (middle) and coordinates' distributions (right).}
\label{fig:chn_chaotic}
\end{figure}


Yet another chaotic dynamical system we consider is the Chen system \cite{chen1999yet} described by the following ODE:
\begin{equation}\label{eq:chen}
\begin{cases}
    \dot{x} &= a(y-x), \\
    \dot{y} &= (c-a)x-xz+cy, \\
    \dot{z} &= xy -bz. 
\end{cases}
\end{equation}
where $a, b, c>0$ are parameters.
When $a=35, b=3, c=28$, 
the system \eqref{eq:chen} has a double-scroll chaotic attractor often observed from a physical, electronic chaotic circuit. The true parameter that we will infer is $u^\dagger=(a^\dagger,b^\dagger,c^\dagger)=(35,3,28)$.
With $u^\dagger$, the system has three unstable equilibrium states given by $(0, 0, 0)$, $(\gamma, \gamma, 2c-a)$, and $(-\gamma, -\gamma, 2c-a)$ where $\gamma = \sqrt{b(2c-a)}$ \cite{yassen2003chaos}.
Figure \ref{fig:chn_chaotic} illustrates the two-scroll attractor (left), the chaotic trajectories (middle) and their marginal and pairwise distributions (right) of their coordinates viewed as random variables.

The Chen dynamics has trajectories changing rapidly as the Lorenz63 dynamics (compare the middle panels of Figure \ref{fig:chn_chaotic} and Figure \ref{fig:lrz_chaotic}). Therefore we adopt the same spin-up length ($t_0=100$) and observation window size ($T=10$) as in the Lorenz inverse problem (Section \ref{sec:Lorenz}). We generate the spatiotemporal data and the augmented time-averaged summary data by observing the trajectory of \eqref{eq:chen} over $[t_0, t_0+T]$ solved with $u^\dagger$ similarly as in the previous sections.
A log-Nomral prior is adopted for $u$: $\log u \sim \mN(\mu_0, \sigma_0^2)$ with $\mu_0=(3.5, 1.2, 3.3)$ and $\sigma_0=(0.35, 0.5, 0.15)$.
The STGP model \eqref{eq:STGP} still posses more convex posterior density $p(u)$ than the time-averaged model \eqref{eq:time-average} as illustrated by its marginal and pairwise sections plotted in Figure \ref{fig:chn_pairpdf}.

\begin{figure}[t]
\includegraphics[width=1\textwidth,height=.5\textwidth]{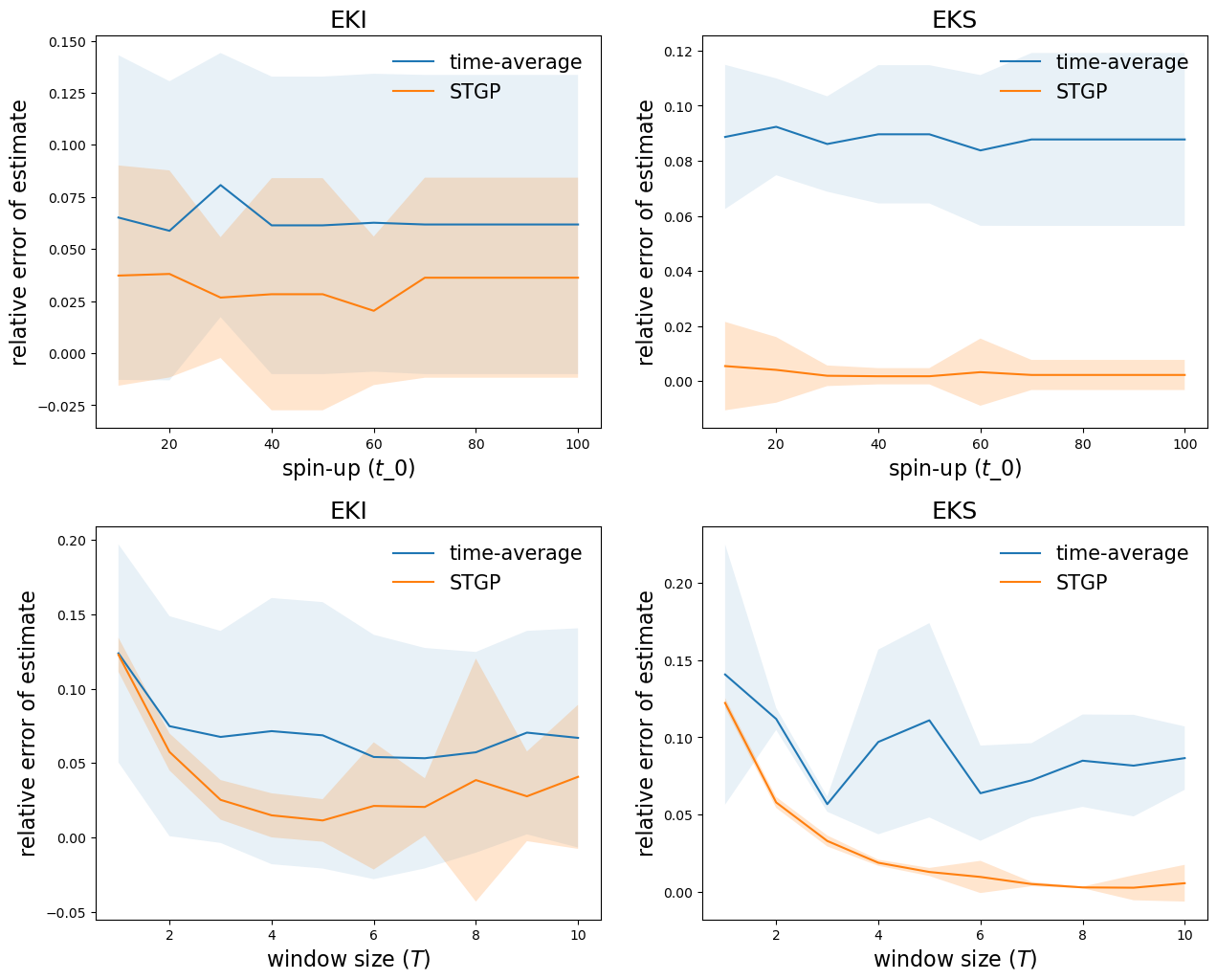}
\caption{Chen inverse problem: comparing posterior estimates of parameter $u$ for two models (time-average and STGP) in terms of relative error of mean $\mathrm{REM}=\frac{\Vert \hat u - u^\dagger \Vert}{\Vert u^\dagger \Vert}$. The upper row shows the results by varying the spin-up $t_0$ and fixing $T=10$. The lower row shows the results by varying the observation window size  $T$ and fixing $t_0=100$. Each experiment is repeated for 10 runs of EnK (EKI and EKS respectively) with $J=500$ ensembles and the shaded regions indicate standard deviations of such repeated experiments.}
\label{fig:chn_rem_spinavgs}
\end{figure}

Varying the spin-up length $t_0$ and the observation window size $T$ one at a time in Figure \ref{fig:chn_rem_spinavgs}, we observe similar advantage of the STGP model compared with the time-averaged model regardless of the insensitivity of errors with respect to $t_0$.
Similarly, the STGP model demands a smaller observation window than the time-averaged model ($T=2$ vs $T=6$ with EKI and $T=2$ vs $T=3$ with EKS) to reach the same level of accuracy.

\begin{table}[ht]\scriptsize
\centering
\begin{tabular}{l|lllll}
\toprule
Model-Algo &        J=50 &       J=100 &       J=200 &               J=500 &              J=1000 \\
\midrule
      Tavg-EKI & 0.07 (0.03) & 0.04 (0.04) & 0.04 (0.04) &         0.05 (0.04) &         0.04 (0.04) \\
      Tavg-EKS & 0.12 (0.03) & 0.10 (0.02) & 0.09 (0.02) &         0.09 (0.01) &         0.09 (0.01) \\
       STGP-EKI & 0.14 (0.09) & 0.09 (0.08) & 0.09 (0.08) &         0.03 (0.03) &     0.01 (9.87e-03) \\
       STGP-EKS & 0.07 (0.04) & 0.05 (0.04) & 0.01 (0.01) & 2.89e-03 (6.07e-03) & 3.32e-04 (4.66e-04) \\
\bottomrule
\end{tabular}
\caption{Chen inverse problem: comparing posterior estimates of parameter $u$ for two models (time-average and STGP) in terms of relative error of median $\mathrm{REM}=\frac{\Vert \hat u - u^\dagger \Vert}{\Vert u^\dagger \Vert}$. Each experiment is repeated for 10 runs of EnK (EKI and EKS respectively) and the numbers in the bracket are standard deviations of such repeated experiments.}
\label{tab:chn_rem}
\end{table}

\begin{figure}[htbp]
\includegraphics[width=.49\textwidth,height=.4\textwidth]{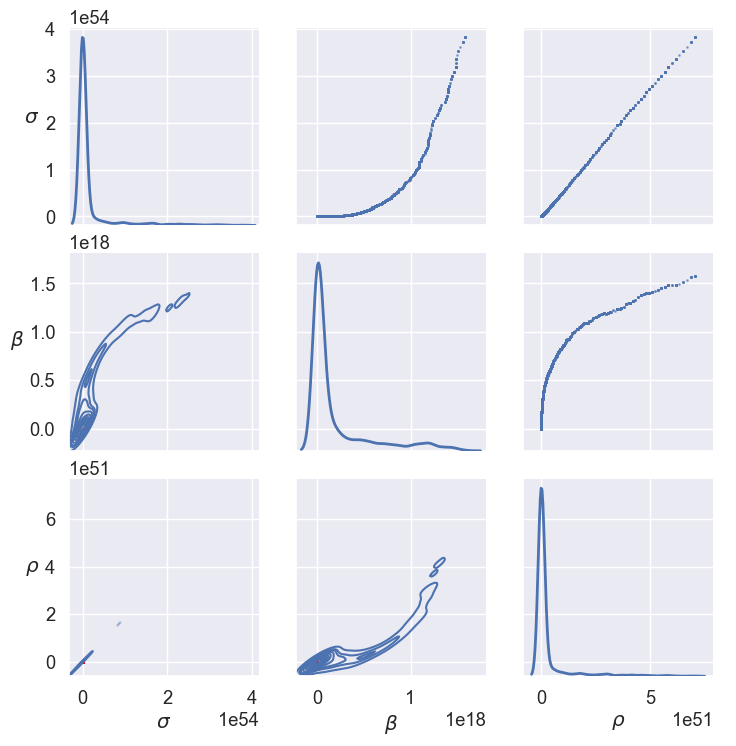}
\includegraphics[width=.49\textwidth,height=.4\textwidth]{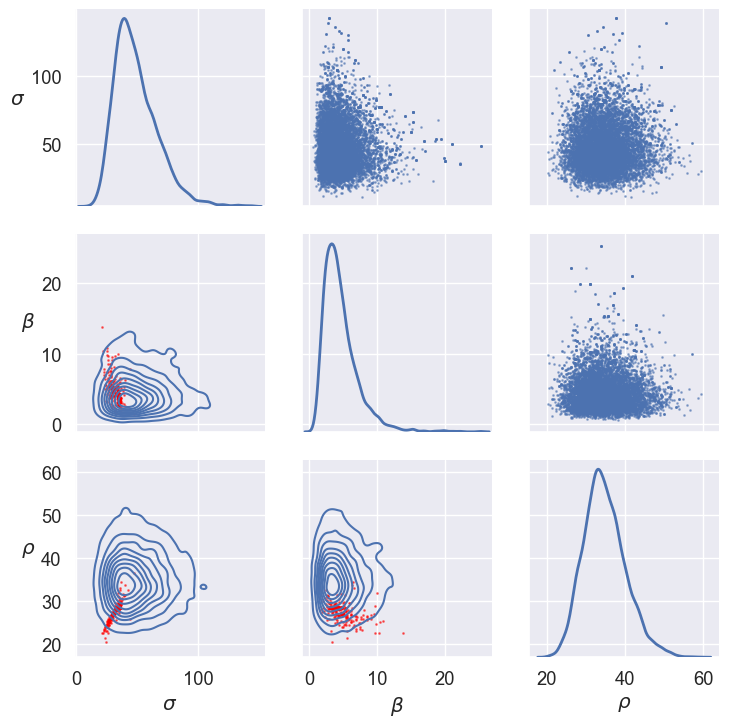}
\caption{Chen inverse problem: marginal (diagonal) and pairwise (lower triangle) distributions estimated with 10000 samples (upper triangle) by the pCN algorithm based on NN emulators for the time-averaged model (left) and the STGP model (right) respectively. Red dots (lower triangle) are selective 10000 ensemble particles from running the EKS algorithm.}
\label{fig:chn_UQ}
\end{figure}

Again we see the merit of the STGP model \eqref{eq:STGP} in reducing the error (REM) of parameter estimation compared with the time-averaged model \eqref{eq:time-average} in various combinations of EnK algorithms with different ensemble sizes ($J$) in Figure \ref{fig:chn_rem} and Table \ref{tab:chn_rem}. As in the previous problem (Section \ref{sec:Rossler}), similar over-fitting (bottom left of Figure \ref{fig:chn_rem}) by the time-averaged model occurs if running EKS algorithms more than 5 iterations (or earlier).

UQ results (Figure \ref{fig:chn_UQ}) by CES show the STGP model estimates the uncertainty of parameter $u$ more appropriately than the time-averaged model.
Finally, though the prediction is challenging to the Chen dynamics \eqref{eq:chen}, the STGP model still performs much better than the time-averaged model by predicting more accurate trajectory for longer time ($t=111$ vs $t=101$) as shown in Figure \ref{fig:chn_pred_comparelik}.

\begin{figure}[htbp]
\includegraphics[width=1\textwidth,height=.3\textwidth]{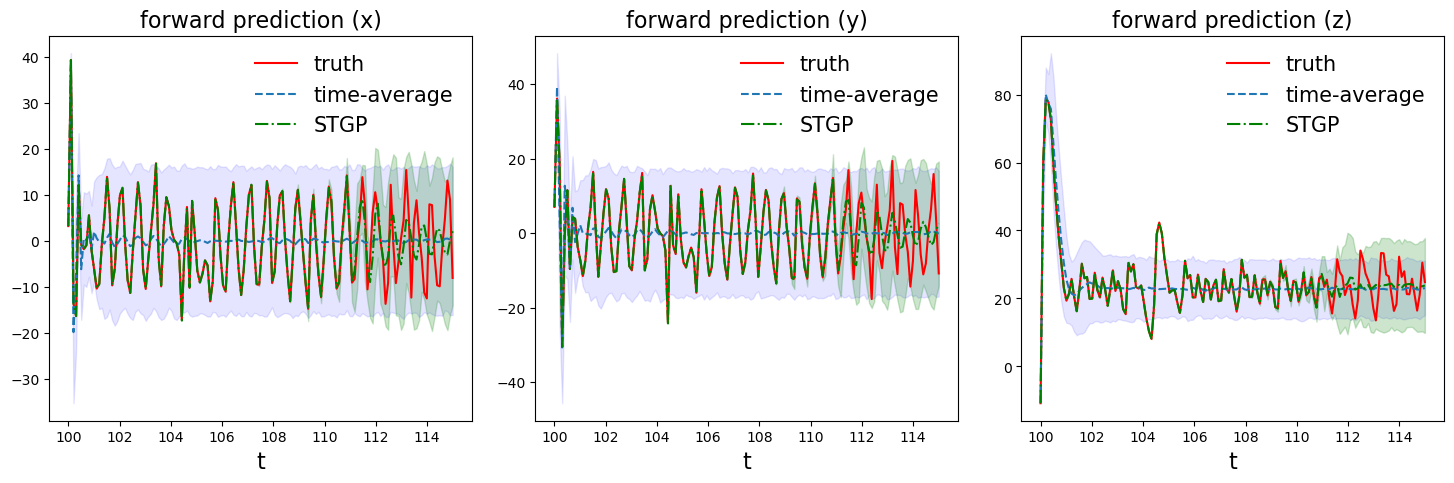}
\caption{Chen inverse problem: comparing forward predictions $\bar{\mG}(\bx, t_*)$ based on the time-averaged model and the STGP model.}
\label{fig:chn_pred_comparelik}
\end{figure}

\section{Conclusion}\label{sec:conclusion}

In this paper, we investigate the inverse problems with spatiotemporal data. We compare the Bayesian models based on STGP with traditional static and time-averaged models that do not fully integrate the spatiotemporal information. By fitting the trajectories of the observed data, the STGP model provides more effective parameter estimation and more appropriate UQ. We explain the superiority of the STGP model in theorems showing that it renders more convex likelihood that facilitates the parameter learning. We demonstrate the advantage of the spatiotemporal modeling using an inverse problem constrained by an advection-diffusion PDE and three inverse problems involving chaotic dynamics.

Theorems \ref{thm:convexity1} and \ref{thm:convexity2} compare the STGP model with the static and the time-averaged models regarding their statistical convexity. These novel qualitative results imply that the parameter learning (based on EnK methods) with the STGP model converges faster than the other two traditional methods. In the future work, we will explore a quantitative characterization on their convergence rates particularly in terms of covariance properties.

The STGP model \eqref{eq:STGP} considered in this paper has a classical separation structure in their joint kernel. This may not be sufficient to characterize complex spatiotemporal relationships, e.g. the temporal evoluation of spatial dependence (TESD) \cite{Lan_2019}. We will expand this work by considering non-stationary non-separable STGP models \cite{cressie2011,Zhang_2020,Wang_2020} to account for more complicated space-time interactions in these spatiotemporal inverse problems.

\section*{Acknowledgement}
SL is supported by NSF grant DMS-2134256.

\clearpage
\section*{Appendix}
\appendix
\newtheorem*{theorem*}{Theorem}

\section{Proofs} \label{apx:proof_convexity}

\begin{theorem*}[\ref{thm:convexity1}]
If we set the maximal eigenvalues of $\bC_\bx$ and $\bC_t$ such that $\lambda_{\max}(\bC_\bx)\lambda_{\max}(\bC_t)\leq \sigma^2_\eps$, then the following inequality holds regarding the Fisher information matrices, $\mI_\textrm{\tiny S}$ and $\mI_\textrm{\tiny ST}$, of the static model and the STGP model respectively:
\begin{equation}
    \mI_\textrm{\tiny ST}(u) \geq \mI_\textrm{\tiny S}(u)
\end{equation}
If we control the maximal eigenvalues of $\bC_\bx$ and $\bC_t$ such that $\lambda_{\max}(\bC_\bx)\lambda_{\max}(\bC_t)\leq J\lambda_{\min}(\Gamma_\textrm{obs})$, then the following inequality holds regarding the Fisher information matrices, $\mI_\textrm{\tiny T}$ and $\mI_\textrm{\tiny ST}$, of the time-averaged model and the STGP model respectively:
\begin{equation}
    \mI_\textrm{\tiny ST}(u) \geq \mI_\textrm{\tiny T}(u)
\end{equation}
\end{theorem*}

\begin{proof}
Denote $\bY_0=\bY-\bM$. We have $\Phi_*(u)=\half \tr\left[\bV_*^{-1} \tp{\bY}_0 \bU_*^{-1} \bY_0\right]$ with $*$ being S or ST. 
$\bU_\textrm{\tiny S}$, $\bV_\textrm{\tiny S}$, $\bU_\textrm{\tiny ST}$ and $\bV_\textrm{\tiny ST}$ are specified in \eqref{eq:MatN}.
We notice that both $\bU_*$ and $\bV_*$ are symmetric, then we have
\begin{equation}
\begin{aligned}
\frac{\pa \Phi_*}{\pa u_i} &= \half\left\{\tr\left[\bV_*^{-1} \frac{\pa \tp{\bY}_0}{\pa u_i} \bU_*^{-1}\bY_0 \right] + \tr\left[\bV_*^{-1} \tp{\bY}_0 \bU_*^{-1}\frac{\pa \bY_0}{\pa u_i} \right] \right\} = \tr\left[\bV_*^{-1} \tp{\bY}_0 \bU_*^{-1}\frac{\pa \bY_0}{\pa u_i} \right] \\
\frac{\pa^2 \Phi_*}{\pa u_i \pa u_j} &= \tr\left[\bV_*^{-1} \tp{\bY}_0 \bU_*^{-1}\frac{\pa^2 \bY_0}{\pa u_i \pa u_j} \right] + \tr\left[\bV_*^{-1} \frac{\pa \tp{\bY}_0}{\pa u_i} \bU_*^{-1}\frac{\pa \bY_0}{\pa u_j} \right] 
\end{aligned}
\end{equation}

Due to the i.i.d. assumption in both models, $\bY_0$ is independent of either $\frac{\pa \bY_0}{\pa u_i}$ or $\frac{\pa^2 \bY_0}{\pa u_i \pa u_j}$. Therefore
\begin{equation}
\begin{aligned}
    (\mI_*)_{ij} &= \E\left[ \frac{\pa^2 \Phi_*}{\pa u_i \pa u_j} \right] = \E \left[ \tr\left(\bV_*^{-1} \frac{\pa \tp{\bY}_0}{\pa u_i} \bU_*^{-1}\frac{\pa \bY_0}{\pa u_j} \right) \right] \\
    &= \E\left[ \VEC\tp{\left(\frac{\pa \bY_0}{\pa u_i}\right)} (\bV_*^{-1} \otimes \bU_*^{-1}) \VEC\left(\frac{\pa \bY_0}{\pa u_j}\right) \right]
\end{aligned}
\end{equation}
For any $\bw=(w_1,\cdots,w_p)\in \mbR^p$ and $\bw\neq \bzero$, denote $\tilde\bw := \sum_{i=1}^p w_i \VEC\left(\frac{\pa \bY_0}{\pa u_i}\right)$.
To prove $\mI_\textrm{\tiny ST}(u) \geq \mI_\textrm{\tiny S}(u)$, it suffices to show $\tp{\tilde \bw}(\bV_\textrm{\tiny ST} \otimes \bU_\textrm{\tiny ST})^{-1} \tilde \bw \geq \tp{\tilde \bw}(\bV_\textrm{\tiny S} \otimes \bU_\textrm{\tiny S})^{-1} \tilde \bw$.

By \cite[Theorem 4.2.12 in][]{Horn_1991}, we know that any eigenvalue of $\bV_* \otimes \bU_*$ has the format as a product of eigenvalues of $\bV_*$ and $\bU_*$ respectively, i.e. $\lambda_k(\bV_* \otimes \bU_*) = \lambda_i(\bV_*)\lambda_j(\bU_*)$, where where $\{\lambda_j(M)\}$ are the ordered eigenvalues of $M$, i.e. $\lambda_1(M)\geq\cdots\geq \lambda_d(M)$.
By the given condition we have
\begin{equation}
    \lambda_{IJ}((\bV_\textrm{\tiny ST} \otimes \bU_\textrm{\tiny ST})^{-1}) = \lambda_1^{-1}(\bV_\textrm{\tiny ST} \otimes \bU_\textrm{\tiny ST}) = \lambda_1^{-1}(\bC_t) \lambda_1^{-1}(\bC_\bx) \geq \sigma^{-2}_\eps = \lambda_1((\bV_\textrm{\tiny S} \otimes \bU_\textrm{\tiny S})^{-1})
\end{equation}
Thus it completes the proof of the first inequality.

Similarly by the second condition, we have
\begin{equation}
    \lambda_{IJ}((\bV_\textrm{\tiny ST} \otimes \bU_\textrm{\tiny ST})^{-1}) = \lambda_1^{-1}(\bC_t) \lambda_1^{-1}(\bC_\bx) \geq J^{-1} \lambda_{\min}^{-1} (\Gamma_\textrm{obs}) = \lambda_1(\bV_\textrm{\tiny S}^- \otimes \bU_\textrm{\tiny S}^{-1})
\end{equation}
and complete the proof of the second inequality.
\end{proof}

\begin{theorem*}[\ref{thm:convexity2}]
If we choose $\bC_\bx=\Gamma_\textrm{obs}$ and require the maximal eigenvalue of $\bC_t$, $\lambda_{\max}(\bC_t)\leq J$, then the following inequality holds regarding the Fisher information matrices, $\mI_\textrm{\tiny T}$ and $\mI_\textrm{\tiny ST}$, of the time-averaged model and the STGP model respectively:
\begin{equation}
    \mI_\textrm{\tiny ST}(u) \geq \mI_\textrm{\tiny T}(u)
\end{equation}
\end{theorem*}

\begin{proof}
Denote $\bY_0=\bY-\bM$. We have $\Phi_*(u)=\half \tr\left[\bV_*^{-1} \tp{\bY}_0 \bU_*^{-1} \bY_0\right]$ with $*$ being T or ST. 
$\bU_\textrm{\tiny T}$, $\bV_\textrm{\tiny T}$, $\bU_\textrm{\tiny ST}$ and $\bV_\textrm{\tiny ST}$ are specified in \eqref{eq:MatN}.

By the similar argument of the proof in Theorem \ref{thm:convexity1}, we have
\begin{equation}
    (\mI_*)_{ij} = \E\left[ \frac{\pa^2 \Phi_*}{\pa u_i \pa u_j} \right] = \tr\left[\bV_*^{-1} \E\left(\frac{\pa \tp{\bY}_0}{\pa u_i} \bU_*^{-1}\frac{\pa \bY_0}{\pa u_j} \right)\right]
\end{equation}
For any $\bw=(w_1,\cdots,w_p)\in \mbR^p$ and $\bw\neq \bzero$, denote $\bW := \sum_{i,j=1}^p w_i \E\left(\frac{\pa \tp{\bY}_0}{\pa u_i} \bU_*^{-1}\frac{\pa \bY_0}{\pa u_j} \right) w_j$.
We know $\bW\geq \bzero_{J\times J}$.
It suffices to show $\tr[\bV_\textrm{\tiny ST}^{-1}\bW]\geq \tr[\bV_\textrm{\tiny T}^{-1}\bW]$.

By the corollary \cite{Marshall_2011} of Von Neumann's trace inequality \cite{Mirsky1975}, we have
\begin{equation}
    \sum_{j=1}^J\lambda_j(\bV_*^{-1})\lambda_{J-j+1}(\bW) \leq \tr(\bV_*^{-1}\bW) \leq \sum_{j=1}^J\lambda_j(\bV_*^{-1})\lambda_j(\bW)
\end{equation}
where $\{\lambda_j(M)\}$ are the ordered eigenvalues of $M$, i.e. $\lambda_1(M)\geq\cdots\geq \lambda_d(M)$.
The only non-zero eigenvalue of $\bV_{\tiny T}^-=J^{-2} (\bm{1}_J \tp{\bm{1}}_J)$ is $\lambda_1(\bV_{\tiny T}^-)=J^{-1}$. Therefore, we have
\begin{equation}
    \tr[\bV_{\tiny T}^-\bW] \leq J^{-1}\lambda_1(\bW) \leq \lambda_J(\bV_\textrm{\tiny ST}^{-1}) \lambda_1(\bW) + \sum_{j=1}^{J-1}\lambda_j(\bV_\textrm{\tiny ST}^{-1})\lambda_{J-j+1}(\bW) \leq \tr[\bV_\textrm{\tiny ST}^{-1}\bW]
\end{equation}
where $\lambda_J(\bV_\textrm{\tiny ST}^{-1}) = \lambda_1^{-1}(\bC_t)\geq J^{-1}$ and $\lambda_j(\bV_\textrm{\tiny ST}^{-1}), \lambda_j(\bW)\geq 0$.
\end{proof}

\section{More Numerical Results}

\begin{figure}[t]
\includegraphics[width=1\textwidth,height=.4\textwidth]{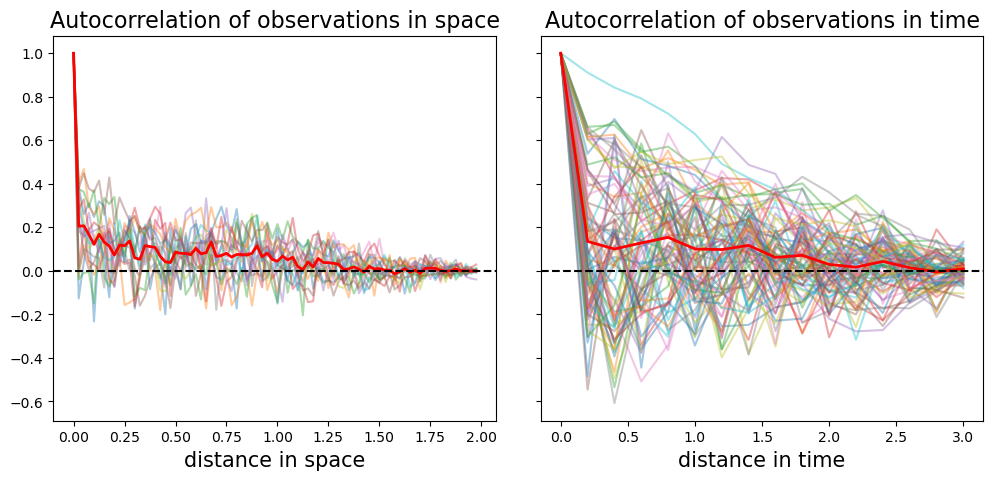}
\caption{Advection-diffusion inverse problem: auto-correlations of observations in space (left) and time (right) respectively.}
\label{fig:adif_obs_acf}
\end{figure}

\begin{figure}[t]
\includegraphics[width=1\textwidth,height=.5\textwidth]{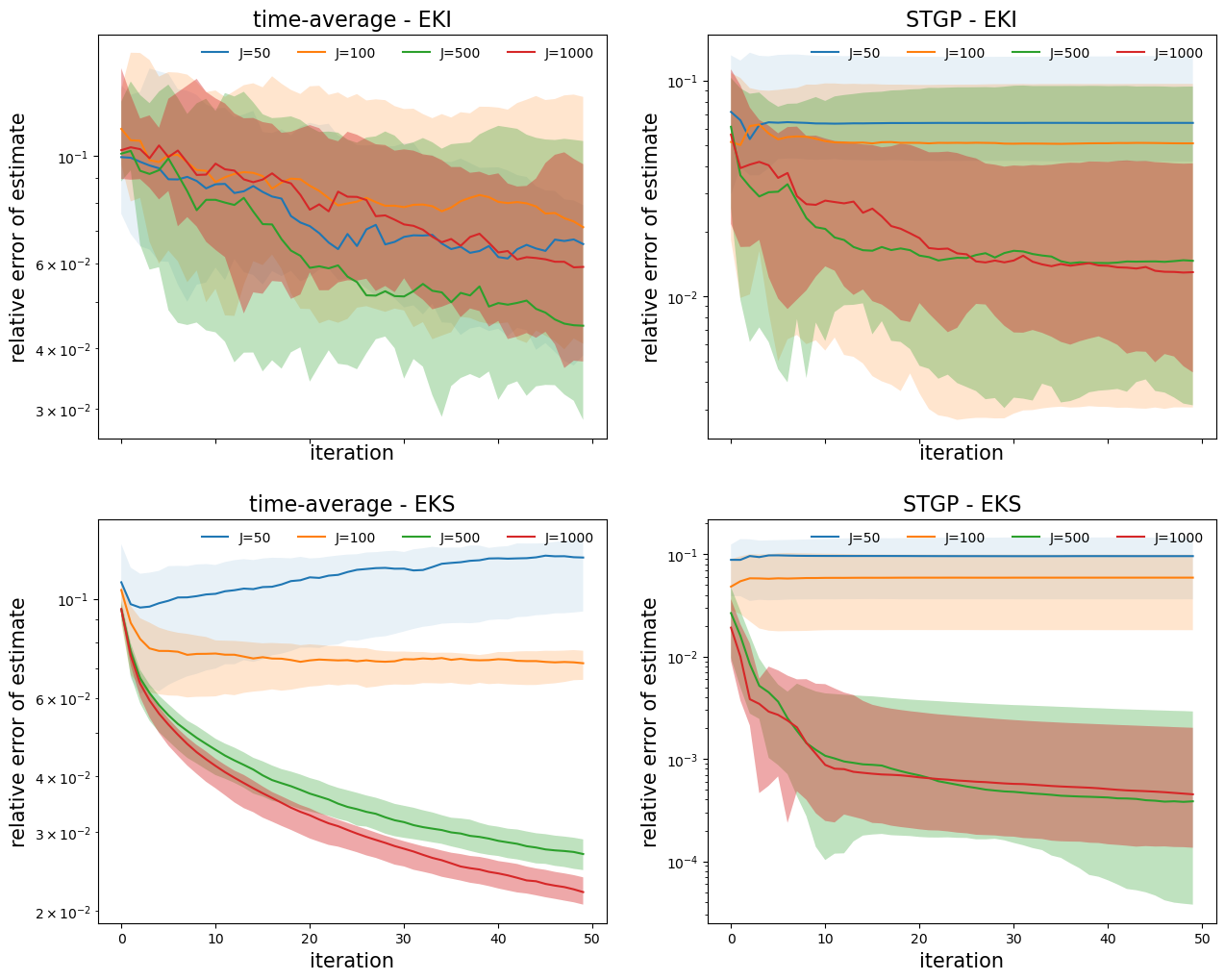}
\caption{Lorenz inverse problem: comparing posterior estimates of parameter $u$ for two models (time-average and STGP) in terms of relative error of median $\mathrm{REM}=\frac{\Vert \hat u - u^\dagger \Vert}{\Vert u^\dagger \Vert}$. Each experiment is repeated for 10 runs of EnK (EKI and EKS respectively) and shaded regions indicate $5\sim 95\%$ quantiles of such repeated results.}
\label{fig:lrz_rem}
\end{figure}


\begin{figure}[t]
\includegraphics[width=.49\textwidth,height=.4\textwidth]{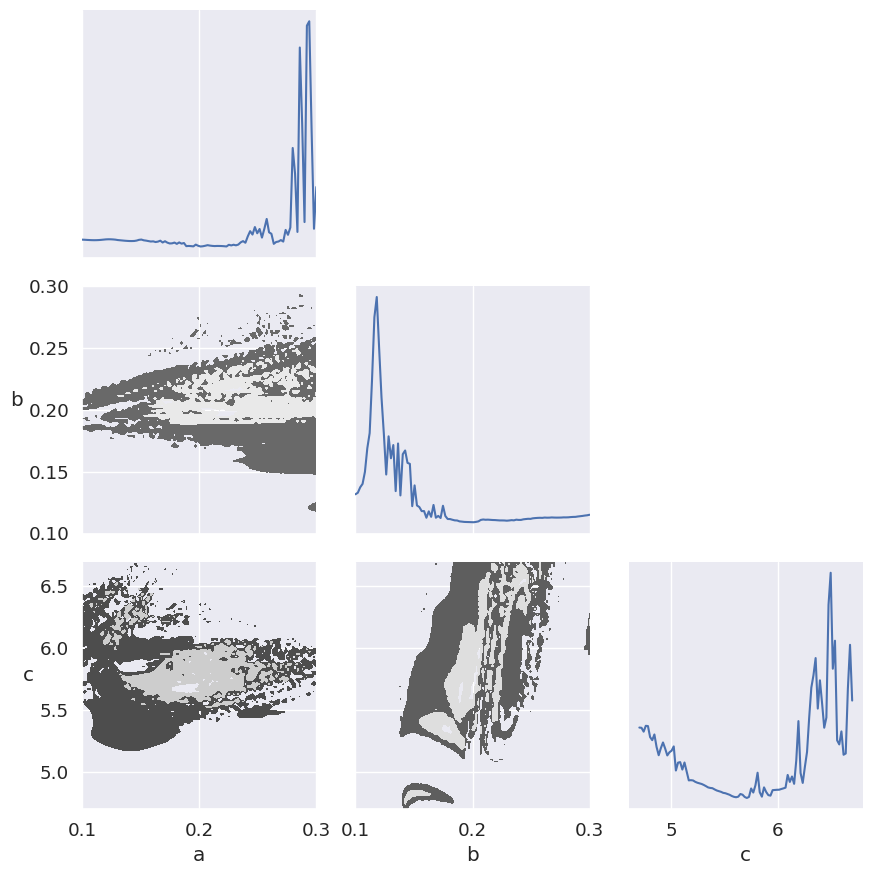}
\includegraphics[width=.49\textwidth,height=.4\textwidth]{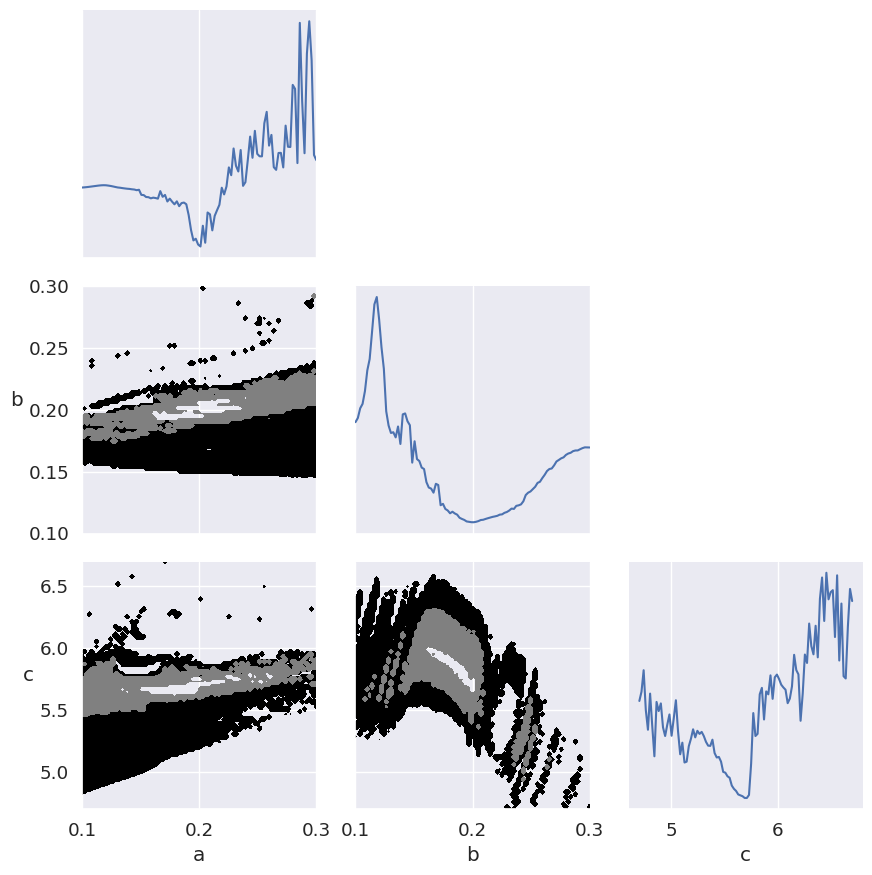}
\caption{R\"ossler inverse problem: marginal (diagonal) and pairwise (lower triangle) sections of the joint density $p(u)$ by the time-averaged model (left) and the STGP model (right) respectively.}
\label{fig:rsl_pairpdf}
\end{figure}

\begin{figure}[t]
\includegraphics[width=1\textwidth,height=.5\textwidth]{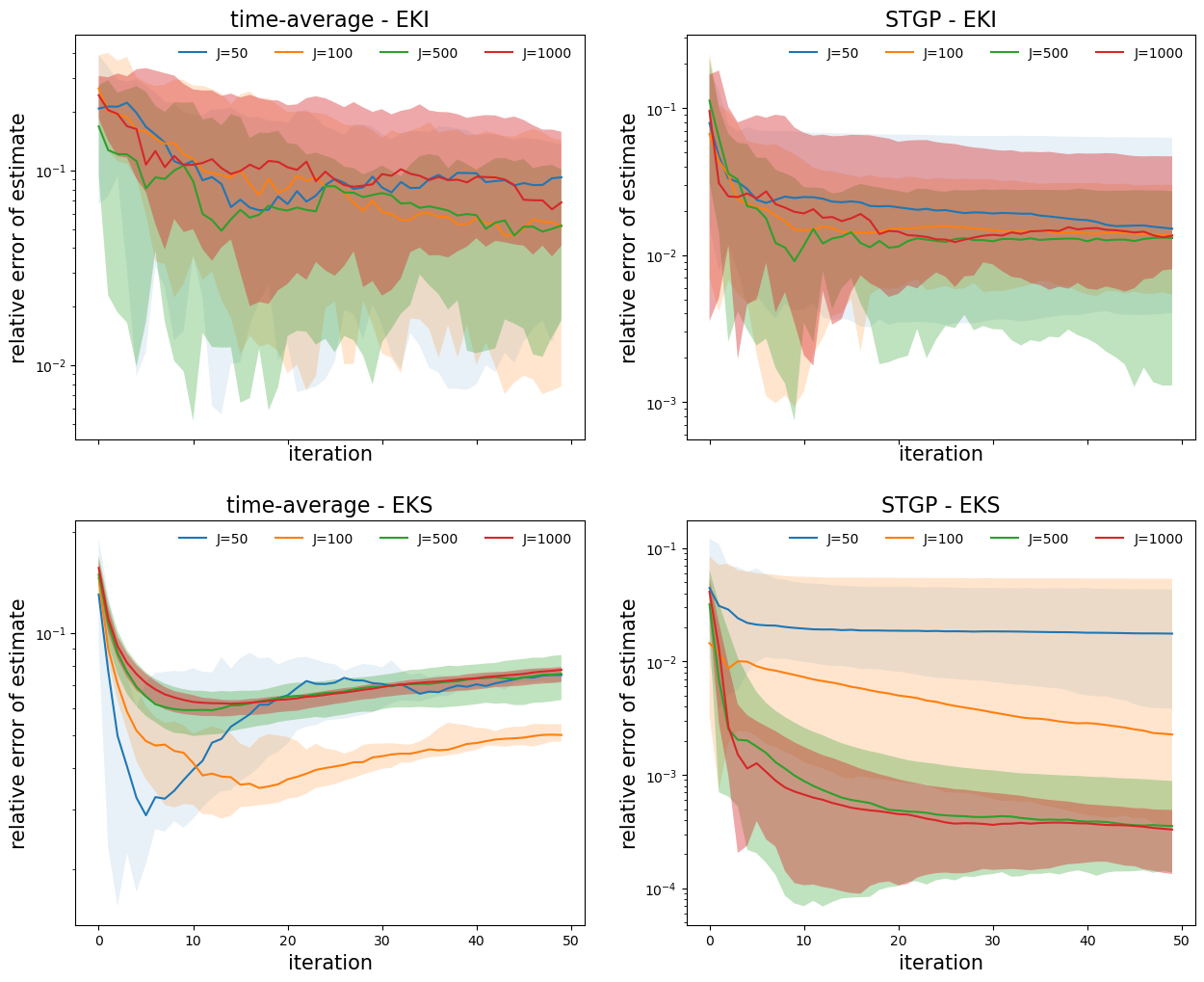}
\caption{R\"ossler inverse problem: comparing posterior estimates of parameter $u$ for two models (time-average and STGP) in terms of relative error of median $\mathrm{REM}=\frac{\Vert \hat u - u^\dagger \Vert}{\Vert u^\dagger \Vert}$. Each experiment is repeated for 10 runs of EnK (EKI and EKS respectively) and shaded regions indicate $5\sim 95\%$ quantiles of such repeated results.}
\label{fig:rsl_rem}
\end{figure}

\begin{figure}[t]
\includegraphics[width=.49\textwidth,height=.4\textwidth]{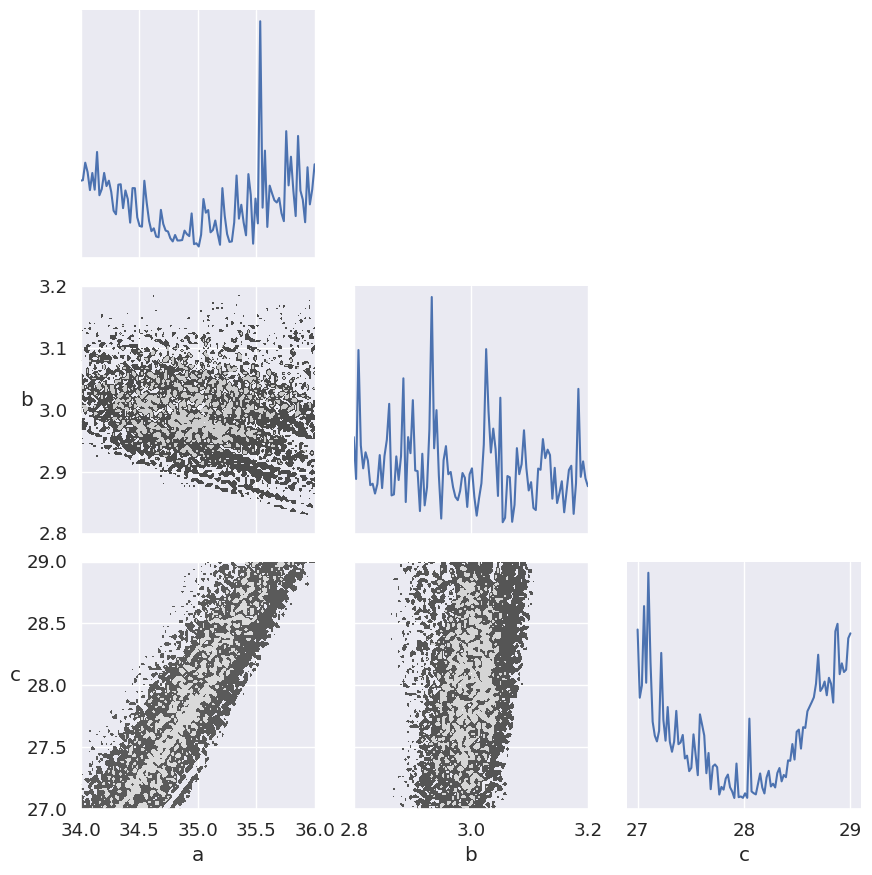}
\includegraphics[width=.49\textwidth,height=.4\textwidth]{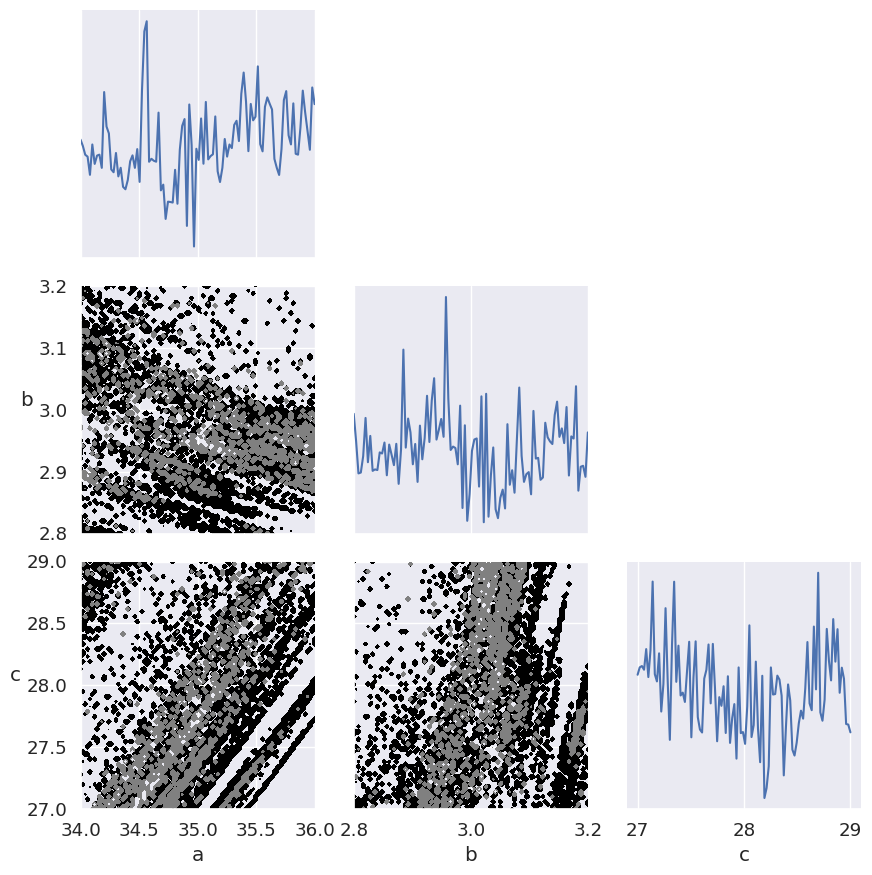}
\caption{Chen inverse problem: marginal (diagonal) and pairwise (lower triangle) sections of the joint density $p(u)$ by the time-averaged model (left) and the STGP model (right) respectively.}
\label{fig:chn_pairpdf}
\end{figure}

\begin{figure}[t]
\includegraphics[width=1\textwidth,height=.5\textwidth]{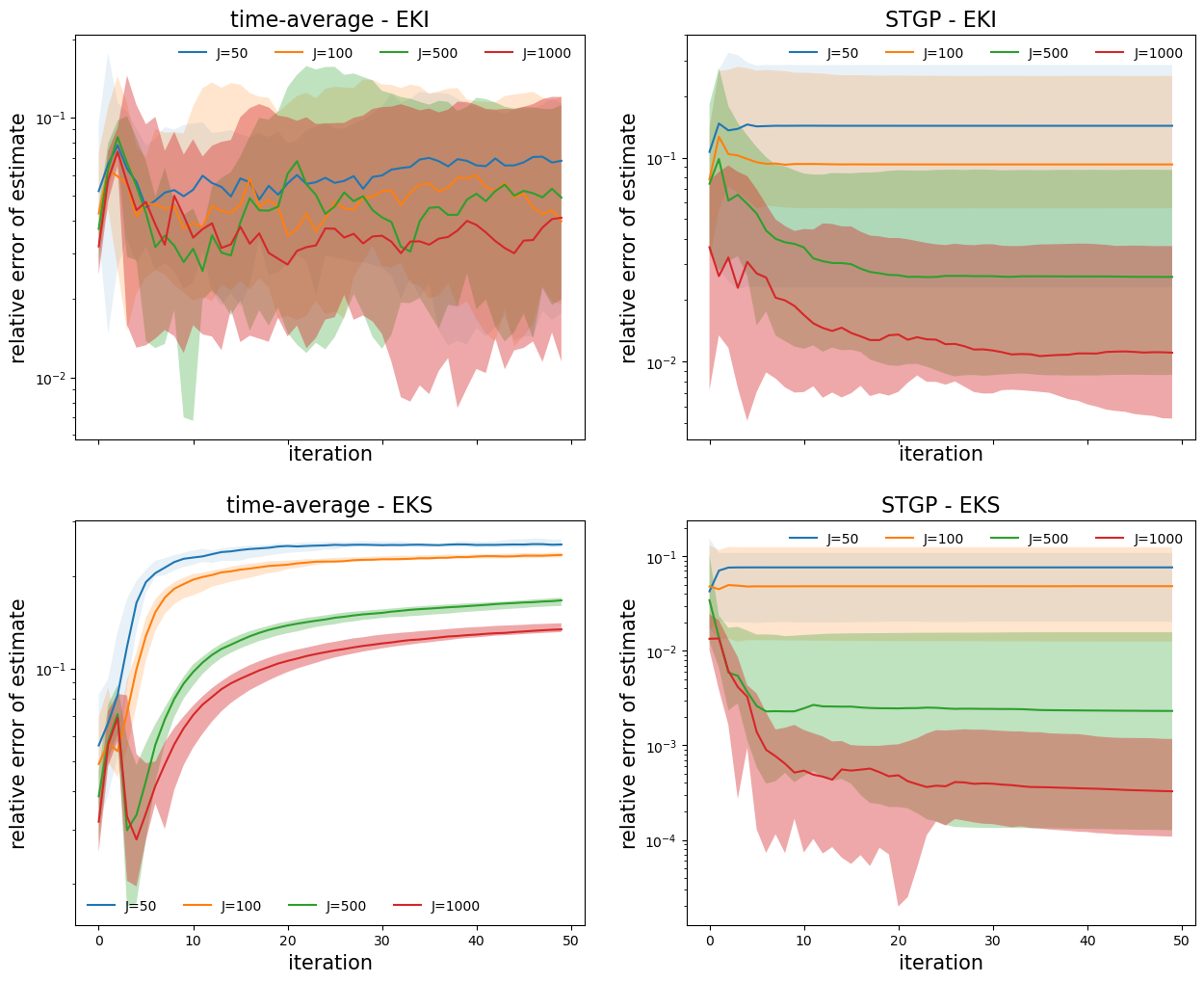}
\caption{Chen inverse problem: comparing posterior estimates of parameter $u$ for two models (time-average and STGP) in terms of relative error of median $\mathrm{REM}=\frac{\Vert \hat u - u^\dagger \Vert}{\Vert u^\dagger \Vert}$. Each experiment is repeated for 10 runs of EnK (EKI and EKS respectively) and shaded regions indicate $5\sim 95\%$ quantiles of such repeated results.}
\label{fig:chn_rem}
\end{figure}

\clearpage
\bibliographystyle{plain}
 \bibliography{references}

\begin{thebibliography}{10}

\bibitem{Abarbanel_2013}
Henry Abarbanel.
\newblock {\em Predicting the Future: completing models of observed complex
  systems}, volume~1.
\newblock Springer New York, 2013.

\bibitem{agiza2001synchronization}
HN~Agiza and MT~Yassen.
\newblock Synchronization of {R}ossler and {C}hen chaotic dynamical systems
  using active control.
\newblock {\em Physics Letters A}, 278(4):191--197, 2001.

\bibitem{beskos11}
A.~Beskos, F.~J. Pinski, J.~M. Sanz-Serna, and A.~M. Stuart.
\newblock {Hybrid Monte-Carlo on Hilbert spaces}.
\newblock {\em Stochastic Processes and their Applications}, 121:2201--2230,
  2011.

\bibitem{beskos14}
Alexandros Beskos.
\newblock A stable manifold {MCMC} method for high dimensions.
\newblock {\em Statistics \& Probability Letters}, 90:46--52, 2014.

\bibitem{beskos2017}
Alexandros Beskos, Mark Girolami, Shiwei Lan, Patrick~E. Farrell, and Andrew~M.
  Stuart.
\newblock {Geometric MCMC for infinite-dimensional inverse problems}.
\newblock {\em Journal of Computational Physics}, 335, 2017.

\bibitem{beskos08}
Alexandros Beskos, Gareth Roberts, Andrew Stuart, and Jochen Voss.
\newblock {MCMC} methods for diffusion bridges.
\newblock {\em Stochastics and Dynamics}, 8(03):319--350, 2008.

\bibitem{Bishop_2017}
Robert Bishop.
\newblock {Chaos}.
\newblock In Edward~N. Zalta, editor, {\em The {Stanford} Encyclopedia of
  Philosophy}. Metaphysics Research Lab, Stanford University, {S}pring 2017
  edition, 2017.

\bibitem{Brooks_1998}
Chris Brooks.
\newblock Chaos in foreign exchange markets: a sceptical view.
\newblock {\em Computational Economics}, 11(3):265--281, 1998.

\bibitem{chada2019}
Neil~K. Chada, Andrew~M. Stuart, and Xin~T. Tong.
\newblock Tikhonov regularization within ensemble kalman inversion, 2019.

\bibitem{chen1999yet}
Guanrong Chen and Tetsushi Ueta.
\newblock Yet another chaotic attractor.
\newblock {\em International Journal of Bifurcation and chaos},
  9(07):1465--1466, 1999.

\bibitem{cleary2020}
Emmet Cleary, Alfredo Garbuno-Inigo, Shiwei Lan, Tapio Schneider, and Andrew~M
  Stuart.
\newblock Calibrate, emulate, sample, 2020.

\bibitem{Conjard_2021}
Maxime Conjard and Henning Omre.
\newblock Spatio-temporal inversion using the selection kalman model.
\newblock {\em Frontiers in Applied Mathematics and Statistics}, 7, apr 2021.

\bibitem{cotter13}
Simon~L Cotter, Gareth~O Roberts, AM~Stuart, and David White.
\newblock {MCMC} methods for functions: modifying old algorithms to make them
  faster.
\newblock {\em Statistical Science}, 28(3):424--446, 2013.

\bibitem{cressie2011}
N.~Cressie and C.K. Wikle.
\newblock {\em Statistics for Spatio-Temporal Data}.
\newblock CourseSmart Series. Wiley, 2011.

\bibitem{dashti2017}
Masoumeh Dashti and Andrew~M. Stuart.
\newblock {\em The {B}ayesian {A}pproach to {I}nverse {P}roblems}, pages
  311--428.
\newblock Springer International Publishing, Cham, 2017.

\bibitem{deWiljes2018}
J.~de~Wiljes, S.~Reich, and W.~Stannat.
\newblock Long-time stability and accuracy of the ensemble kalman--bucy filter
  for fully observed processes and small measurement noise.
\newblock {\em SIAM Journal on Applied Dynamical Systems}, 17(2):1152--1181,
  2018.

\bibitem{Echeverria_2009}
David Echeverr{\'\i}a~Ciaurri and T.~Mukerji.
\newblock A robust scheme for spatio-temporal inverse modeling of oil
  reservoirs.
\newblock 01 2009.

\bibitem{effah2018study}
S~Effah-Poku, William Obeng-Denteh, and IK~Dontwi.
\newblock A study of chaos in dynamical systems.
\newblock {\em Journal of Mathematics}, 2018, 2018.

\bibitem{evensen1994}
Geir Evensen.
\newblock Sequential data assimilation with a nonlinear quasi-geostrophic model
  using monte carlo methods to forecast error statistics.
\newblock {\em Journal of Geophysical Research}, 99(C5):10143, 1994.

\bibitem{Evensen_1996}
Geir Evensen and Peter~Jan van Leeuwen.
\newblock Assimilation of geosat altimeter data for the agulhas current using
  the ensemble kalman filter with a quasigeostrophic model.
\newblock {\em Monthly Weather Review}, 124(1):85--96, jan 1996.

\bibitem{Fisher_1922}
R.~A. Fisher and Edward~John Russell.
\newblock On the mathematical foundations of theoretical statistics.
\newblock {\em Philosophical Transactions of the Royal Society of London.
  Series A, Containing Papers of a Mathematical or Physical Character},
  222(594-604):309--368, 1922.

\bibitem{Garbuno-Inigo_2020}
Alfredo Garbuno-Inigo, Franca Hoffmann, Wuchen Li, and Andrew~M. Stuart.
\newblock Interacting langevin diffusions: Gradient structure and ensemble
  kalman sampler.
\newblock {\em SIAM Journal on Applied Dynamical Systems}, 19(1):412--441,
  2020.

\bibitem{Garbuno_Inigo_2020}
Alfredo Garbuno-Inigo, Nikolas N{\"u}sken, and Sebastian Reich.
\newblock Affine invariant interacting langevin dynamics for bayesian
  inference.
\newblock {\em {SIAM} Journal on Applied Dynamical Systems}, 19(3):1633--1658,
  jan 2020.

\bibitem{Gupta_2018}
A.K. Gupta and D.K. Nagar.
\newblock {\em Matrix Variate Distributions}, chapter Chapter 2: MATRIX VARIATE
  NORMAL DISTRIBUTION.
\newblock Chapman and Hall/{CRC}, may 2018.

\bibitem{hegazi2001controlling}
A~Hegazi, HN~Agiza, and MM~El-Dessoky.
\newblock Controlling chaotic behaviour for spin generator and {R}ossler
  dynamical systems with feedback control.
\newblock {\em Chaos, Solitons \& Fractals}, 12(4):631--658, 2001.

\bibitem{Horn_1991}
Roger~A. Horn and Charles~R. Johnson.
\newblock {\em Topics in Matrix Analysis}.
\newblock Cambridge University Press, apr 1991.

\bibitem{Huang_2022}
Daniel~Zhengyu Huang, Jiaoyang Huang, Sebastian Reich, and Andrew~M. Stuart.
\newblock Efficient derivative-free bayesian inference for large-scale inverse
  problems, 2022.

\bibitem{Iglesias_2016}
Marco~A Iglesias.
\newblock A regularizing iterative ensemble kalman method for {PDE}-constrained
  inverse problems.
\newblock {\em Inverse Problems}, 32(2):025002, Jan 2016.

\bibitem{Iglesias_2013}
Marco~A Iglesias, Kody J~H Law, and Andrew~M Stuart.
\newblock Ensemble kalman methods for inverse problems.
\newblock {\em Inverse Problems}, 29(4):045001, Mar 2013.

\bibitem{Ivancevic_2008}
Vladimir~G. Ivancevic and Tijana~T. Ivancevic.
\newblock {\em Complex Nonlinearity}.
\newblock Springer Berlin Heidelberg, 2008.

\bibitem{Baukal_2000}
Charles E.~Baukal Jr., Vladimir Gershtein, and Xianming~Jimmy Li, editors.
\newblock {\em Computational Fluid Dynamics in Industrial Combustion}.
\newblock {CRC} Press, oct 2000.

\bibitem{LAN2019a}
Shiwei Lan.
\newblock Adaptive dimension reduction to accelerate infinite-dimensional
  geometric markov chain monte carlo.
\newblock {\em Journal of Computational Physics}, 392:71 -- 95, September 2019.

\bibitem{Lan_2019}
Shiwei Lan.
\newblock Learning temporal evolution of spatial dependence with generalized
  spatiotemporal gaussian process models.
\newblock arXiv:1901.04030, 08 2021.

\bibitem{lan2022}
Shiwei Lan, Shuyi Li, and Babak Shahbaba.
\newblock Scaling up bayesian uncertainty quantification for inverse problems
  using deep neural networks.
\newblock {\em SIAM/ASA Journal on Uncertainty Quantification}, to appear,
  2022.

\bibitem{Liz_2012}
Eduardo Liz and Alfonso Ruiz-Herrera.
\newblock Chaos in discrete structured population models.
\newblock {\em {SIAM} Journal on Applied Dynamical Systems}, 11(4):1200--1214,
  jan 2012.

\bibitem{Long_2011}
Christopher~J. Long, Patrick~L. Purdon, Simona Temereanca, Neil~U. Desai,
  Matti~S. H{\"a}m{\"a}l{\"a}inen, and Emery~N. Brown.
\newblock State-space solutions to the dynamic magnetoencephalography inverse
  problem using high performance computing.
\newblock {\em The Annals of Applied Statistics}, 5(2B), jun 2011.

\bibitem{Lorenz_1963}
Edward~N. Lorenz.
\newblock Deterministic nonperiodic flow.
\newblock {\em Journal of the Atmospheric Sciences}, 20(2):130--141, mar 1963.

\bibitem{Marshall_2011}
Albert~W. Marshall, Ingram Olkin, and Barry~C. Arnold.
\newblock {\em Inequalities: Theory of Majorization and Its Applications}.
\newblock Springer New York, 2nd edition, 2011.

\bibitem{Mirsky1975}
L.~Mirsky.
\newblock A trace inequality of john von neumann.
\newblock {\em Monatshefte f{\"u}r Mathematik}, 79:303--306, 1975.

\bibitem{Morzfeld_2018}
M.~Morzfeld, J.~Adams, S.~Lunderman, and R.~Orozco.
\newblock Feature-based data assimilation in geophysics.
\newblock {\em Nonlinear Processes in Geophysics}, 25(2):355--374, 2018.

\bibitem{Negrini_2021}
Elisa Negrini, Giovanna Citti, and Luca Capogna.
\newblock System identification through lipschitz regularized deep neural
  networks.
\newblock {\em Journal of Computational Physics}, 444:110549, nov 2021.

\bibitem{Ojeda_2019}
+Alejandro Ojeda, +Marius Klug, +Kenneth Kreutz-Delgado, +Klaus Gramann, and
  +Jyoti Mishra.
\newblock A bayesian framework for unifying data cleaning, source separation
  and imaging of electroencephalographic signals.
\newblock {\em bioRxiv}, 2019.

\bibitem{ott1981strange}
Edward Ott.
\newblock Strange attractors and chaotic motions of dynamical systems.
\newblock {\em Reviews of Modern Physics}, 53(4):655, 1981.

\bibitem{Pasha_2021}
Mirjeta Pasha, Arvind~K. Saibaba, Silvia Gazzola, Malena~I. Espanol, and Eric
  de~Sturler.
\newblock Efficient edge-preserving methods for dynamic inverse problems, 2021.

\bibitem{Petra2011}
N.~Petra and G.~Stadler.
\newblock Model variational inverse problems governed by partial differential
  equations.
\newblock Technical report, The Institute for Computational Engineering and
  Sciences, The University of Texas at Austin., 2011.

\bibitem{ROSSLER_1976}
O.E. R{\"o}ssler.
\newblock An equation for continuous chaos.
\newblock {\em Physics Letters A}, 57(5):397--398, 1976.

\bibitem{ROSSLER_1979}
O.E. Rossler.
\newblock An equation for hyperchaos.
\newblock {\em Physics Letters A}, 71(2):155--157, 1979.

\bibitem{Schillings_2017a}
C.~Schillings and A.~Stuart.
\newblock Analysis of the ensemble kalman filter for inverse problems.
\newblock {\em SIAM Journal on Numerical Analysis}, 55(3):1264--1290, 2017.

\bibitem{Schillings_2017b}
C.~Schillings and A.~M. Stuart.
\newblock Convergence analysis of ensemble kalman inversion: the linear, noisy
  case.
\newblock {\em Applicable Analysis}, 97(1):107--123, Oct 2017.

\bibitem{Schneider2017}
Tapio Schneider, Shiwei Lan, Andrew Stuart, and Jo{\~a}o Teixeira.
\newblock Earth system modeling 2.0: A blueprint for models that learn from
  observations and targeted high-resolution simulations.
\newblock {\em Geophysical Research Letters}, 44(24):12,396--12,417, 2017.

\bibitem{Shcherbakova_2021}
A~I Shcherbakova, Y~A Kupriyanova, and G~V Zhikhareva.
\newblock Spatio-temporal analysis the results of solving the inverse problem
  of electrocardiography.
\newblock {\em Journal of Physics: Conference Series}, 2091(1):012028, nov
  2021.

\bibitem{SIREGAR_1996}
Pridi Siregar and Jean-Paul Sinteff.
\newblock Introducing spatio-temporal reasoning into the inverse problem in
  electroencephalography.
\newblock {\em Artificial Intelligence in Medicine}, 8(2):97--122, 1996.

\bibitem{stuart10}
Andrew~M Stuart.
\newblock Inverse problems: a {B}ayesian perspective.
\newblock {\em Acta Numerica}, 19:451--559, 2010.

\bibitem{villa2020}
Umberto Villa, Noemi Petra, and Omar Ghattas.
\newblock {hIPPYlib}: An extensible software framework for large-scale inverse
  problems governed by {PDE}s; part i: {D}eterministic {I}nversion and
  {L}inearized {B}ayesian {I}nference, 2020.

\bibitem{Wang_2020}
Kangrui Wang, Oliver Hamelijnck, Theodoros Damoulas, and Mark Steel.
\newblock Non-separable non-stationary random fields.
\newblock In Hal~Daum{\'e} III and Aarti Singh, editors, {\em Proceedings of
  the 37th International Conference on Machine Learning}, volume 119 of {\em
  Proceedings of Machine Learning Research}, pages 9887--9897. PMLR, 13--18 Jul
  2020.

\bibitem{Woolrich_2004}
M.W. Woolrich, M.~Jenkinson, J.M. Brady, and S.M. Smith.
\newblock Fully bayesian spatio-temporal modeling of {FMRI} data.
\newblock {\em {IEEE} Transactions on Medical Imaging}, 23(2):213--231, feb
  2004.

\bibitem{yang2002control}
Shyi-Kae Yang, Chieh-Li Chen, and Her-Terng Yau.
\newblock Control of chaos in {L}orenz system.
\newblock {\em Chaos, Solitons \& Fractals}, 13(4):767--780, 2002.

\bibitem{Yang_2017}
Ying Yang.
\newblock Source-space analyses in meg/eeg and applications to explore
  spatio-temporal neural dynamics in human vision.
\newblock 2017.

\bibitem{Yao_2016}
Bing Yao and Hui Yang.
\newblock Physics-driven spatiotemporal regularization for high-dimensional
  predictive modeling: A novel approach to solve the inverse {ECG} problem.
\newblock {\em Scientific Reports}, 6(1), dec 2016.

\bibitem{yassen2003chaos}
MT~Yassen.
\newblock Chaos control of {C}hen chaotic dynamical system.
\newblock {\em Chaos, Solitons \& Fractals}, 15(2):271--283, 2003.

\bibitem{Zhang_2020}
Bohai Zhang and Noel Cressie.
\newblock Bayesian inference of spatio-temporal changes of arctic sea ice.
\newblock {\em Bayesian Analysis}, 15(2):605--631, jun 2020.

\bibitem{Zhang_2005}
Yiheng Zhang, Alireza Ghodrati, and Dana~H Brooks.
\newblock An analytical comparison of three spatio-temporal regularization
  methods for dynamic linear inverse problems in a common statistical
  framework.
\newblock {\em Inverse Problems}, 21(1):357--382, jan 2005.

\end{thebibliography}
 
\end{document}